\documentclass[11pt]{article}
\usepackage{amssymb, amsmath, amsthm}

\usepackage{float}
\usepackage[english]{babel}
\usepackage[utf8]{inputenc}
\usepackage{textcomp}
\usepackage[hidelinks]{hyperref}
\usepackage[T1]{fontenc}
\usepackage{comment}
\usepackage{todonotes}
\usepackage[capitalise]{cleveref}
\usepackage{fullpage}
\usepackage{mathtools}
\usepackage{caption}
\usepackage{multirow}
\usepackage[ruled,vlined,linesnumbered]{algorithm2e}
\usepackage{wrapfig}
\usetikzlibrary{patterns}
\usepackage{csquotes}

\usepackage[font={small}]{caption}

\usepackage{hyperref}
\hypersetup{
        colorlinks,
    linkcolor={red!50!black},
    citecolor={blue!50!black},
    urlcolor={blue!80!black}
}

\newcommand{\R}{\mathbb{R}}
\newcommand{\N}{\mathbb{N}}
\newcommand{\Z}{\mathbb{Z}}

\newcommand{\mydef}{:=}

\newcommand{\domcomp}{n^3\,\text{polylog}\, n}

\DeclareMathOperator{\dist}{dist}
\DeclareMathOperator{\Int}{Int}
\DeclareMathOperator{\area}{area}

\newcommand\drop[1]{}

\newcommand{\set}[1]{\left\{ #1 \right\}}
\newcommand{\setbuilder}[2]{\set{ #1 \; \middle\vert \; #2 }}

\makeatletter
\newtheorem*{rep@theorem}{\rep@title}
\newcommand{\newreptheorem}[2]{%
\newenvironment{rep#1}[1]{%
 \def\rep@title{#2 \ref{##1}}%
 \begin{rep@theorem}}%
 {\end{rep@theorem}}}
\makeatother

\newtheorem{theorem}{Theorem}
\newreptheorem{theorem}{Theorem}

\newtheorem{lemma}[theorem]{Lemma}
\newreptheorem{lemma}{Lemma}
\newtheorem{corollary}[theorem]{Corollary}

\theoremstyle{definition}
\newtheorem{definition}{Definition}
\theoremstyle{remark}
\newtheorem*{remark}{Remark}
\newtheorem*{claim}{Claim}

\title{Tiling with Squares and Packing Dominos in Polynomial Time}
\author{Anders Aamand\footnote{Basic Algorithms Research Copenhagen (BARC), University of Copenhagen. BARC is supported by the VILLUM Foundation grant 16582.} \and Mikkel Abrahamsen$^*$ \and Thomas D. Ahle$^*$  \and Peter M. R. Rasmussen$^*$}
\date{August 9, 2021}

\begin{document}
\maketitle
\thispagestyle{empty}
\begin{abstract}
A polyomino is a polygonal region with axis parallel edges and corners of integral coordinates, which may have holes.
In this paper, we consider planar tiling and packing problems with polyomino pieces and a polyomino container $P$.
We give two polynomial time algorithms, one for deciding if $P$ can be tiled with $k\times k$ squares for any fixed $k$ which can be part of the input (that is, deciding if $P$ is the union of a set of non-overlapping $k\times k$ squares) and one for packing $P$ with a maximum number of non-overlapping and axis-parallel $2\times 1$ dominos, allowing rotations by $90^\circ$.
As packing is more general than tiling, the latter algorithm can also be used to decide if $P$ can be tiled by $2\times 1$ dominos.

These are classical problems with important applications in VLSI design, and the related problem of finding a maximum packing of $2\times 2$ squares is known to be NP-Hard [J.~Algorithms 1990].
For our three problems there are known pseudo-polynomial time algorithms, that is, algorithms with running times polynomial in the \emph{area} or \emph{perimeter} of $P$.
However, the standard, compact way to represent a polygon is by listing the coordinates of the corners in binary.
We use this representation, and thus present the first polynomial time algorithms for the problems.
Concretely, we give a simple $O(n\log n)$ algorithm for tiling with squares, and a more involved $O(\domcomp)$ algorithm for packing and tiling with dominos, where $n$ is the number of corners of $P$.
\end{abstract}

\newpage
\pagenumbering{arabic} 

\section{Introduction}
\begin{wrapfigure}{r}{0.35\textwidth}
   \centering
   \begin{tikzpicture}[scale=.6]
       \foreach \x in {0,...,7}
       \foreach \y in {0,...,7} {
          \pgfmathparse{int(1-mod(\x+\y,2) + int(\x==0 && \y==0) + int(\x==7 && \y==7))}
          \ifnum\pgfmathresult=1 {
            \fill[pattern=north east lines, pattern color=black] (\x,\y) rectangle (\x+1, \y+1);
          } \fi
        }
        \draw (1,1) -- (1,0) -- (8,0) -- (8,7) -- (7,7) -- (7,8) -- (0,8) -- (0,1) -- (1,1); 
   \end{tikzpicture}
   \caption{The chessboard polyomino envisioned by Max Black.}
   \label{fig:chessboard}
\end{wrapfigure}
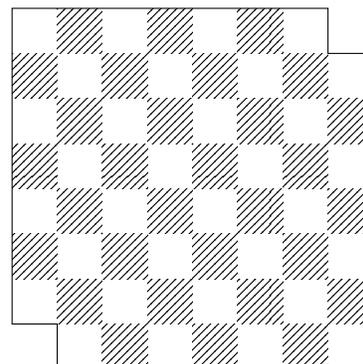

A chessboard has been mutilated by removing two diagonally opposite corners, leaving $62$ squares.
Philosopher Max Black asked in 1946 whether one can place $31$ dominoes of size $1\times 2$ so as to cover all of the remaining squares?
Tiling problems of this sort are popular in recreational mathematics, such as the mathematical olympiads\footnote{See e.g.~the ``hook problem'' of the International Mathematical Olympiad 2004.} and have been discussed by Golomb~\cite{doi:10.1080/00029890.1954.11988548} and Gamow \& Stern~\cite{gamow1958puzzle}.
The mutilated chessboard and the dominos are examples of the type of polygon called a \emph{polyomino}, which is a polygonal region of the plane with axis parallel edges and corners of integral coordinates.
We allow polyominos to have holes.

From an algorithmic point of view, it is natural to ask whether a given (large) polyomino $P$ can be \emph{tiled} by copies of another fixed (small) polyomino $Q$, which means that $P$ is the union of non-overlapping copies of $Q$ that may or may not be rotated by $90^\circ$ and $180^\circ$.
As the answer is often a boring \emph{no}, one can ask more generally for the largest number of copies of $Q$ that can be \emph{packed} into the given container $P$ without overlapping.
Algorithms answering this question (for various $Q$) turn out to have important applications in very large scale integration (VLSI) circuit technology.
As a concrete example, Hochbaum \& Maass~\cite{hochbaum1985approximation} gave the following motivation for their development of a polynomial time approximation scheme for packing $2\times 2$ squares into a given polyomino $P$ (using the area representation of $P$, to be defined later):
\begin{displayquote}
``For example, $64$K RAM chips, some of which may be defective, are available on a rectilinear grid placed on a silicon wafer.
$2 \times 2$ arrays of such nondefective chips could be wired together to produce $256$K RAM chips.
In order to maximize yield, we want to pack a maximal number of such $2 \times 2$ arrays into the array of working chips on a wafer.''
\end{displayquote}
Although the mentioned amounts of memory are small compared to those of present day technology, the basic principles behind the production of computer memory are largely unchanged, and methods for circumventing defective cells of wafers (the cells are also known as \emph{dies} in this context) is still an active area of research in semiconductor manufacturing~\cite{chien2001cutting,de2005investigation,JangWafer,melzner2007maximization}.

\vspace{.5em}
The most important result in tiling is perhaps the combinatorial group theory approach by Conway \& Lagarias~\cite{CONWAY1990183}.
Their algorithmic technique is used to decide whether a given finite region consisting of cells in a regular lattice (triangular, square, or hexagonal) can be tiled by pieces drawn from a finite set of tile shapes.
Thurston~\cite{thurston1990conway} gives a nice introduction to the technique and shows how it can be used to decide if a polyomino without holes can be tiled by dominos.
The running time is $O(a\log a)$, where $a$ is the \emph{area} of $P$.
Pak, Sheffer, \& Tassy~\cite{pak2016fast} described an algorithm with running time $O(p\log p)$, where $p$ is the \emph{perimeter} of $P$.

The problem of \emph{packing} a maximum number of dominos into a given polyomino $P$ 
was apparently first analyzed by
Berman, Leighton, \& Snyder~\cite{berman1982optimal} who observed that this problem can be reduced to finding a maximum matching of the incidence graph $G(P)$ of the cells in $P$:
There is a vertex for each $1\times 1$ cell in $P$, and two vertices are connected if the two cells share a geometrical edge.
The graph $G(P)$ is bipartite, so the problem can be solved in $O(n^{3/2})$ time using the Hopcroft–Karp algorithm, where $n$ is the number of cells (i.e., the area of~$P$).

On the flip-side, a number of hardness results have been obtained for simple tiling and packing problems:
Beauquier, Nivat, Remila, \& Robson~\cite{BEAUQUIER19951} showed that if $P$ can have holes, the problem of deciding if $P$ can be tiled by translates of two rectangles $1\times m$ and $k\times 1$ is NP-complete as soon as $\max\{m,k\}\geq 3$ and $\min\{m,k\}\geq 2$.
Pak \& Yang~\cite{PAK20131804} showed that there exists a set of at most $10^6$ rectangles such that deciding whether a given \emph{hole-free} polyomino can be tiled with translates from the set is NP-complete.
Other generalizations have even turned out be undecidable:
Berger~\cite{bergerUndecidability} proved in 1966 that deciding whether pieces from a given finite set of polyominos can tile the plane is Turing complete.
For packing, Fowler, Paterson, \& Tanimoto~\cite{fowler1981optimal} showed already in the early 80s that deciding whether a given number of $3\times 3$ squares can be packed into a polyomino (with holes) is NP-complete, and the result was strengthened to $2\times 2$ squares by
Berman, Johnson, Leighton, Shor, \& Snyder~\cite{BERMAN1990153}.

\vspace{.5em}

As it turns out, for all of the above results, it is assumed that the container $P$ is represented either as a list of the individual cells forming the interior of $P$ or as a list of the boundary cells.
We shall call these representations the \emph{area representation} and \emph{perimeter representation}, respectively.
The area and perimeter representations correspond to a unary rather than binary representation of integers and the running times of the existing algorithms are thus only pseudo-polynomial.
It is much more efficient and compact to represent $P$ by the coordinates of the corners, where the coordinates are represented as binary numbers. 
This is the way one would usually represent polygons (with holes) in computational geometry:
The corners are given in cyclic order as they appear on the boundary of $P$, one cycle for the outer boundary and one for each of the holes of $P$. We shall call such a representation a \emph{corner representation}.
With a corner representation, the area and perimeter can be exponential in the input size, so the known algorithms which rely on an area or perimeter representation to be polynomial, are in fact exponential when using this more efficient encoding of the input.
Problems that are NP-complete in the area or perimeter representation are also NP-hard in the corner representation, but NP-membership does not necessarily follow.
In our practical example of semiconductor manufacturing, the corner representation also seems to be the natural setting for the problem.
Hopefully, there are only few defective cells to be avoided when grouping the chips, so the total number of corners of the usable region is much smaller than its area.

El-Khechen, Dulieu, Iacono, \& Van Omme~\cite{el2009packing} showed that even using a corner representation for a polymino $P$, the problem of deciding if $m$ squares of size $2\times 2$ can be packed into $P$ is in NP.
That was not clear before since the naive certificate  specifies the placement of each of the $m$ squares, and so, would have exponential length.
Beyond this, we know of no other work using the corner representation for polyomino tiling or packing problems.

\paragraph{Our contribution.}



While the complexity of the problem of packing $2\times 2$ squares into a polyomino $P$ has thus been settled as NP-complete, the complexity of the tiling problem was left unsettled.
Tiling and packing are closely connected in this area of geometry, but their complexities can be drastically different.
Indeed, we show in Section~\ref{sec:tiling} that it can be decided in $O(n\log n)$ time by a surprisingly simple algorithm whether $P$ can be tiled by $k\times k$ squares for any fixed $k\in\mathbb N$ which can even be part of the input.
Here, $n$ is the number of corners of $P$.\footnote{We assume throughout the paper that we can make basic operations (additions, subtractions, comparisons) on the coordinates in $O(1)$ time.
Otherwise, the time complexities of our two algorithms will be $O(n t \log n)$ and $O(n^3t+\domcomp)$, respectively, where $t$ is the time it takes to make one such operation.}
With the area and perimeter representations, it is trivial to decide if $P$ can be tiled in polynomial time (see Section~\ref{sec:tiling}), but as noted above, using the corner representation, it is not even immediately obvious that the problem is in NP.

In Section~\ref{sec:algo}-\ref{sec:runningtime}, we provide and analyse an algorithm that can decide in $O(\domcomp)$ time if $m$ dominos (i.e., rectangles of size $1\times 2$ that can be rotated $90^\circ$) can be packed in a given polyomino $P$.
This algorithm is more complicated and we consider it our most important contribution.
The algorithm implicitly constructs a maximum packing, and the same algorithm can be used to decide if $P$ can be tiled by dominos.
In Section~\ref{sec:simpler}, we further describe a much simpler algorithm that works by truncating long edges of $P$ (using a multiple-sink multiple-source maximum flow algorithm as a black box). The simplicity of this algorithm comes at the cost of a higher running time of $O(n^4\,\text{polylog}\, n)$.
The proof that the simple algorithm works uses the same structural results as we developed for the faster but more complicated algorithm.
Table~\ref{tab:four-problems} summarises the known and new results.

\bgroup
\def\arraystretch{2}
\begin{table}[ht]
\centering
\begin{tabular}{ c | c | c }
 Shapes & Tiling & Packing \\ 
 \hline
   \begin{tikzpicture}[scale=.3, baseline=-.05em]
       \foreach \x in {0,...,2}
           \draw (\x,0) -- (\x,1); 
       \foreach \y in {0,...,1}
           \draw (0,\y) -- (2,\y); 
   \end{tikzpicture}
   \;
   \begin{tikzpicture}[scale=.3, baseline=.35em]
       \foreach \x in {0,...,1}
           \draw (\x,0) -- (\x,2); 
       \foreach \y in {0,...,2}
           \draw (0,\y) -- (1,\y); 
   \end{tikzpicture}
 & $O(\domcomp)$ [This paper] & $O(\domcomp)$ [This paper] \\  
 \hline
   \begin{tikzpicture}[scale=.3, baseline=.4em]
       \foreach \x in {0,...,2}
           \draw (\x,0) -- (\x,2); 
       \foreach \y in {0,...,2}
           \draw (0,\y) -- (2,\y); 
   \end{tikzpicture}
 & $O(n\log n)$ [This paper] & NP-complete~\cite{BERMAN1990153,el2009packing}
\end{tabular}
\caption{Complexities of the four fundamental tiling and packing problems.
Here, $n$ is the number of corners of the container $P$.
The algorithm for tiling with squares works for any size $k\times k$.}
\label{tab:four-problems}
\end{table}
\egroup

\paragraph{Further related work.}

The technique by Conway \& Lagarias~\cite{CONWAY1990183} has been adapted to obtain algorithms for tiling with other shapes than dominos: Kenyon \& Kenyon~\cite{kenyontiling} showed how to decide whether a given hole-free polyomino $P$ can be tiled with translates of the rectangles $1\times m$ and $k\times 1$ for fixed integers $m$ and $k$.
The running time is again linear in the area of $P$.
They also described an algorithm to decide if a polyomino can be tiled by the rectangles $k\times m$ and $m\times k$ with running time quadratic in the area.
R\'{e}mila~\cite{remilatiling} generalized the work by Kenyon and Kenyon and obtained a quadratic time algorithm for deciding whether a given hole-free polyomino can be tiled by translates of two fixed rectangles $k\times m$ and $k'\times m'$.
Wijshoff \& {van Leeuwen}~\cite{WIJSHOFF19841} and Beauquier \& Nivat~\cite{beauquier1991translating} gave algorithms for deciding whether a given polyomino tiles the entire plane.
For work on packing trominos, that is, polyominos consisting of three unit squares, see~\cite{horiyama2012packing}.

\subsection{Our techniques}

\paragraph{Tiling with $k\times k$ squares.}
We sort the corners of the given polyomino $P$ by the $x$-coordinates and use a vertical sweep-line $\ell$ that sweeps over $P$ from left to right.
The intuition is that the algorithm keeps track of how the tiling looks in the region of $P$ to the left of $\ell$ if a tiling exists.
As $\ell$ sweeps over $P$, we keep track of how the tiling pattern changes under $\ell$.
Each vertical edge of $P$ that $\ell$ sweeps over causes changes to the tiling, and we must update our data structure accordingly.

\paragraph{Packing with dominos.}
Our basic approach is to reduce the packing problem in the polyomino $P$ (with $n$ corners) to a maximum matching problem in a graph $G^*$ with only $O(n^3)$ vertices and edges.
We prove that a maximum matching in $G^*$ corresponds to a maximum packing of dominos in $P$.
The construction of $G^*$ requires many techniques and the correctness relies on several structural results on domino packings and technical lemmas regarding the particular way we define the intermediate polyominos and graphs that are used to eventually arrive at $G^*$.

We first find the maximum subpolyomino $P_1\subset P$ such that all corners of $P_1$ have even coordinates.
We then use a hole-elimination technique:
By carving channels in $P_1$ from the holes to the boundary, we obtain a hole-free subpolyomino $P_2\subset P_1$.
The particular way we choose the channels is important in order to ensure that the final graph $G^*$ has size only $O(n^3)$.
We now apply a technique of reducing $P$ by removing everything far from the boundary of $P_2$:
We consider the subpolyomino $Q\subset P_2$ of all cells with at least some distance $\Omega(n)$ to the boundary of $P_2$, and then we define $P_3\mydef P\setminus Q$ (note that $Q$ is removed from $P$ and not from $P_2$).
The main insight is that any packing of dominos in $P_3$ can be extended to a packing of all of $P$ that, restricted to $Q$, is a \emph{tiling}.
For this to hold, it turns out to be important that $P_2$ has no holes.

A crucial step is to prove that every cell in the polyomino $P_3$ has distance $O(n)$ to the boundary of $P_3$ and that $P_3$ has $O(n)$ corners.
There may, however, still be an exponential number of cells in $P_3$ due to long \emph{pipes} (corridors).
We then develop a technique for contracting these long pipes.
The contraction is not carried out geometrically, but in the incidence graph $G_3\mydef G(P_3)$ of the cells of $P_3$, by contracting long horizontal and vertical paths to single edges, and the resulting graph is $G^*$.

All vertices of $G^*$ correspond to cells of $P_3$ with distance at most $O(n)$ from a corner of $P_3$, and since $P_3$ has $O(n)$ corners, we get that $G^*$ has size $O(n^3)$.
We then compute a maximum matching in $G^*$ using a multiple-source multiple-sink maximum flow algorithm by Borradaile, Klein, Mozes, Nussbaum, \& Wulff{-}Nilsen~\cite{borradaile2017multiple}, which has since been improved slightly by Gawrychowski \& Karczmarz~\cite{gawrychowski2018improved}.
This results in a running time of $O(\domcomp)$.
The number of dominos in a maximum packing in the original polyomino $P$ is then the size of the maximum matching plus half of the area of everything that has been removed from $P$.

\section{Preliminaries}\label{sec:prelim}

We define a \emph{cell} to be a $1\times 1$ square of the form $[i,i+1]\times [j,j+1]$, $i,j\in \Z$.
A subset $P\subseteq \R^2$ is called a \emph{polyomino} if it is a finite union of cells.
For a polyomino $P$, we define $G(P)$ to be the graph which has the cells in $P$ as vertices and an edge between two cells if they share a (geometrical) edge. We say that $P$ is \emph{connected} if $G(P)$ is a connected graph. Figure~\ref{figure:augmenting} (a) illustrates a connected polyomino. For a simple closed curve $\gamma\subset\R^2$, we denote by $\Int \gamma$ the interior of $\gamma$. An alternative way to represent a connected polyomino is by a sequence of simple closed curves $(\gamma_0,\gamma_1,\dots,\gamma_h)$ such that (1) each of the curves follows the horizontal and vertical lines of the integral grid $\Z^2$, (2) for each $i\in \{1,\dots,h\}$, $\Int \gamma_i \subseteq \Int \gamma_0$, (3) for each distinct $i,j\in \{1,\dots,h\}$, $\Int \gamma_i \cap \Int \gamma_j=\emptyset$, and (4) for distinct $i,j\in \{0,\dots,h\}$, $\gamma_i\cap \gamma_j\subseteq \Z^2$. For a connected polyomino $P$, there exists a unique such sequence (up to permutations of $\gamma_1,\dots,\gamma_h$) with $P=\overline{\Int \gamma}\setminus (\bigcup_{i=1}^h \Int \gamma_i)$. It is standard to reduce our tiling and packing problems to corresponding tiling and packing problems for connected polyominos, so for simplicity we will assume that the input polyominos to our algorithms are connected. The \emph{corners} of a polyomino $P$ (specified by a sequence $(\gamma_0,\gamma_1,\dots,\gamma_h)$), are the corners of the curves $\gamma_0,\dots, \gamma_h$. We assume that an input polyomino with $n$ corners is represented using $O(n)$ words of memory by describing the corners of each of the curves $\gamma_0,\dots\gamma_h$ in cyclic order.

\begin{figure}[ht]
\centering
\includegraphics[page=5]{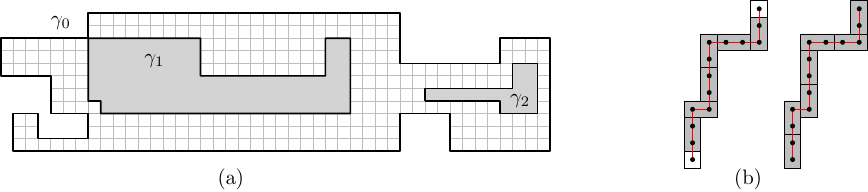}
\caption{(a) A polyomino with two holes. (b) Extending a domino packing using an augmenting path in $G(P)$.}
\label{figure:augmenting}
\end{figure}

In this paper we will exclusively work with the $L_\infty$-norm when measuring distances. For two points $a,b \in \R^2$ we define $\dist(a,b)=\|a-b\|_\infty$. For two subsets $A,B\subseteq \R^2$ we define 
$$
\dist(A,B)=\inf_{(a,b)\in A \times B}\dist(A,B).
$$
In our analysis, $A$ and $B$ will always be closed and bounded (they will in fact be polynomios), and then the $\inf$ can be replaced by a $\min$. Finally, we need the notion of the \emph{offset} $B(A,r)$ of a set $A\subseteq \R^2$ by a value $r\in\R$.
If $r\geq 0$, we define
$$
B(A,r)\mydef \setbuilder{x \in \R^2}{\dist(x,A)\leq r},
$$
and otherwise, we define $B(A,r)\mydef B(A^c,-r)^c$.
Note that if $r\geq 0$, we have $A\subset B(A,r)$ and otherwise, we have $B(A,r)\subset A$.

Note that a domino packing of $P$ naturally corresponds to a matching of $G(P)$ and we will often take this viewpoint. We therefore require some basic matching terminology and a result on how to extend matchings. Let $G$ be a graph and $M$ a matching of $G$. A path $(v_1,\dots,v_{2k})$ of $G$ is said to be an \emph{augmenting path} if $v_1$ and $v_{2k}$ are unmatched in $M$ and for each $1\leq i\leq k-1$, $v_{2i}$ and $v_{{2i}+1}$ are matched to each other in $M$. Modifying $M$ restricted to $\{v_1,\dots,v_{2k}\}$ by instead matching $(v_{2i-1},v_{2i})$ for $1\leq i \leq k$, we obtain a larger matching which now includes the two vertices $v_1$ and $v_{2k}$. See Figure~\ref{figure:augmenting} (b) for an illustration in the context of domino packings. We require the following basic result by Berge which guarantees that any non-maximum matching of $G$ can always be extending to a larger matching using an augmenting path as above.

\begin{lemma}[Berge]\label{lemma:berge}
Let $G$ be a graph and $M$ a matching of $G$ which is not maximum. Then there exists an augmenting path between two unmatched vertices $G$.
\end{lemma}

\section{Tiling with squares}\label{sec:tiling}

\paragraph{Naive algorithm.}
The naive algorithm to decide if $P$ can be tiled with $k\times k$ tiles works as follows.
Consider any convex corner $c$ of $P$.
A $k\times k$ square $S$ must be placed with a corner at $c$.
If $S$ is not contained in $P$, we conclude that $P$ cannot be tiled with $k\times k$ squares.
Otherwise, we recurse of the uncovered part $P\setminus S$.
When nothing is left, we conclude that $P$ can be tiled.
This algorithm runs in time polynomial in the area of $P$ and also shows that if $P$ can be tiled, there is a unique way to do it.

\paragraph{Sweep line algorithm.}
For the ease of presentation, we focus on the case of deciding tileability using $2\times 2$ squares.
It is straightforward to adapt the algorithm to decide tileability by $k\times k$ squares for any fixed $k\in\mathbb N$, as explained in the end of this section.

Our algorithm for deciding if a given polyomino $P$ can be tiled with $2\times 2$ squares uses a vertical sweep line that sweeps over $P$ from left to right.
The intuition is that the algorithm keeps track of how the tiling looks in the region of $P$ to the left of $\ell$ if a tiling exists.
As $\ell$ sweeps over $P$, we keep track of how the tiling pattern changes under $\ell$.
Each vertical edge of $P$ that $\ell$ sweeps over causes changes to the tiling, and we must update our data structures accordingly.

Recall that if $P$ is tileable, then the tiling is unique.
We define $T(P)\subset P$ to be the union of the boundaries of the tiles in the tiling of $P$, i.e., such that $P\setminus T(P)$ is a set of open $2\times 2$ squares.
If $P$ is not tileable, we define $T(P)\mydef\bot$.

Consider the situation where the sweep line is some vertical line $\ell$ with integral $x$-coordinate $x(\ell)$.
The algorithm stores a set $\mathcal I$ of pairwise interior-disjoint closed intervals $\mathcal I=I_1,\ldots,I_m\subset \R$, ordered from below and up.
Each interval $I_i$ has endpoints at integers and represents the segment $I'_i\mydef \{x(\ell)\}\times I_i$ on $\ell$.
In the simple case that no vertical edge of $P$ has $x$-coordinate $x(\ell)$ (so that no change to the set $P\cap \ell$ happens at this point), the intervals $\mathcal I$ together represent the part of $\ell$ in $P$, i.e., we have
$P\cap \ell=\bigcup_{i\in[m]} I'_i$.
If one or more vertical edges of $P$ have $x$-coordinate $x(\ell)$, then $P\cap \ell$ changes at this point and the intervals $\mathcal I$ must be updated accordingly.

For each interval $I_i$ we store a \emph{parity} $p(I_i)\in\{0,1\}$, which encodes how the tiling must be at $I'_i$ if $P$ is tileable.
To make this precise, we state the following \emph{parity invariant} of the algorithm under the assumption that $P$ is tileable; see also Figure~\ref{figure:rectanglepartition}.

\begin{figure}[ht]
\centering
\includegraphics[page=1]{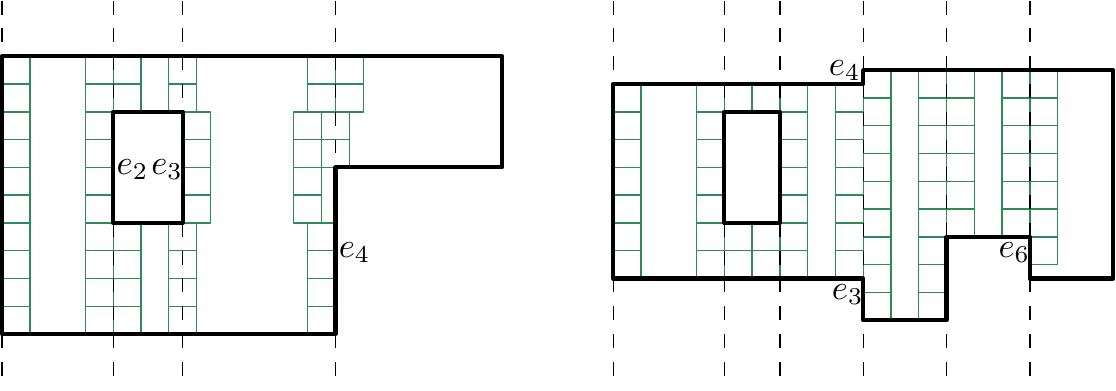}
\caption{Two instances that cannot be tiled. Left: The edge $e_2$ splits the only interval in $\mathcal I$ into two smaller intervals.
Then $e_3$ introduces a new interval with a different parity than the existing two. The edge $e_4$ makes the algorithm conclude that $P$ cannot be tiled since $e_4$ overlaps an interval with the wrong parity.
Right: The edges $e_3$ and $e_4$ introduce new intervals that are merged with the existing one. Edge $e_6$ introduces an interval which is merged with the existing interval and the result has odd length, so the algorithm concludes that $P$ cannot be tiled.}
\label{figure:rectanglepartition}
\end{figure}

\begin{itemize}
\item
If $p(I_i)$ and $x(\ell)$ have the same parity, then $I'_i\subset T(P)$, i.e., $I'_i$ follows the boundaries of some tiles and does not pass through the middle of any tile.

\item
Otherwise, $I'_i\cap T(P)$ consists of isolated points, i.e., $I'_i$ passes through the middle of some of the tiles and does not follow the boundary of any tile.
\end{itemize}

We say that two neighboring intervals $I_i,I_{i+1}$ of $\mathcal I$ are \emph{true} neighbors if $I_i$ and $I_{i+1}$ share an endpoint.
In addition to the parity invariant, we require $\mathcal I$ to satisfy the following \emph{neighbor invariant}: Any pair of true neighbors of $\mathcal I$ have different parity.

The pseudocode of the algorithm is shown in Algorithm~\ref{ALG1}.
Initially, we sort all vertical edges after their $x$-coordinates and break ties arbitrarily.
We then run through the edges in this order.
Each edge makes a change to the set $P\cap \ell$, and we need to update the intervals $\mathcal I$ accordingly so that the parity and the neighbor invariants are satisfied after each edge has been handled.
Figure~\ref{figure:rectanglepartition} shows the various events.

Consider the event that the sweep line $\ell$ reaches a vertical edge $e_j=\{x\}\times [y_0,y_1]$.
If the interior of $P$ is to the left of $e_j$, then $P\cap \ell$ shrinks.
Each interval $I_i\in\mathcal I$ that overlaps $[y_0,y_1]$ must then also shrink, be split into two, or disappear from $\mathcal I$.
This is handled by the for-loop at line~\ref{alg:for2}.
If the parity of one of these intervals $I_i$ does not agree with the parity of $e_j$, we get from the parity invariant that $P$ cannot be tiled, and hence the algorithm returns ``no tiling'' at line~\ref{alg:badoverlap}.

If on the other hand the interior of $P$ is to the right of $e_j$, then $P\cap \ell$ expands and a new interval $I$ must be added to $\mathcal I$.
This is handled by the else-part at line~\ref{alg:else}.
The new interval $I$ may have one or two true neighbors in $\mathcal I$.
If one or two such neighbors also have the same parity as $I$, we merge these intervals into one interval of $\mathcal I$.
This ensures that the neighbor invariant is satisfied after $e_j$ has been handled.

In line~\ref{alg:ifmove}, we consider the case that we finished handling all vertical edges at some specific $x$-coordinate so that the sweep line will move to the right in order to handle the next edge $e_{j+1}$ in the next iteration.
If there is an interval $I_i$ of odd length in $\mathcal I$, it follows from the parity invariant together with the neighbor invariant that $P$ cannot be tiled, so the algorithm returns ``no tiling'' at line~\ref{alg:oddlength}.

\begin{algorithm}[ht]
\LinesNumbered
\DontPrintSemicolon
\SetArgSty{}
\SetKwIF{If}{ElseIf}{Else}{if}{}{else if}{else}{end if}
\SetKwFor{While}{while}{}{end while}
Let $e_1,\ldots,e_k$ be the vertical edges of $P$ in sorted order.\;
\For{$j=1,\ldots,k$}{
Let $[y_0,y_1]$ be the interval of $y$-coordinates of $e_j$.\;
\uIf{the interior of $P$ is to the left of $e_j$}{
\For{each $I_i\in\mathcal I$ that overlaps $[y_0,y_1]$}
{\label{alg:for2}
\If{$I_i$ and $x(e_j)$ have different parity}
{ \label{alg:ifmove0}
\Return {``no tiling''}\; \label{alg:badoverlap}
}
Remove $I_i$ from $\mathcal I$, let
$J\mydef \overline{I_i\setminus [y_0,y_1]}$, and if $J\neq\emptyset$, add the interval(s) in $J$ to $\mathcal I$.
}
}\Else{\label{alg:else}
Make a new interval $I\mydef [y_0,y_1]$ with the parity $p(I)\mydef x(e_j)\text{ mod } 2$ and add $I$ to $\mathcal I$.\; \label{alg:defP}
\If {$I$ has one or two true neighbors in $\mathcal I$ that also have the same parity as $I$}
{
Merge those intervals in $\mathcal I$.\; \label{alg:merge}
}
}
\If{$j<k$ and $x(e_{j+1})>x(e_j)$ and some $I_i\in\mathcal I$ has odd length} { \label{alg:ifmove}
\Return {``no tiling''}\; \label{alg:oddlength}
}
}
\Return {``tileable''}\;
\caption{}
\label{ALG1}
\end{algorithm}

The above explanation of the algorithm argues that if the invariants hold before edge $e_j$ is handled, they also hold after.
It remains to argue that they also hold before the next edge $e_{j+1}$ is handled in the case that the sweep line $\ell$ jumps to the right in order to sweep over $e_{j+1}$.
In the open strip between the vertical lines containing $e_j$ and $e_{j+1}$, there are no vertical segments of $P$.
Hence, the pattern of the tiling $T(P)$ must continue as described by the parities $p(I_i)$ in between the edges $e_j$ and $e_{j+1}$, so the parity invariant also holds before $e_{j+1}$ is handled.

We already argued that if the algorithm returns ``no tiling'', then $P$ is not tileable.
Suppose on the other hand that the algorithm returns ``tileable''.
In order to prove that $P$ can then be tiled, we define for each $j\in[k]$ a polyomino $P_j\subset P$.
We consider the situation where the sweep line $\ell$ contains $e_j$ and $e_j$ has just been handled by the algorithm.
We then define $P_j$ to be the union of
\begin{itemize}
\item the part of $P$ to the left of $\ell$, and
\item the rectangle $[x(\ell),x(\ell)+1]\times I_i$ for each $I_i\in\mathcal I$ with a different parity than $x(\ell)$.
\end{itemize}
We first see that for each $j\in[k]$, we have $P_j\subset P$.
To this end, we just have to check that the rectangles $[x(\ell),x(\ell)+1]\times I_i$ are in $P$.
If one such interval was not in $P$, there would be an edge of $P$ overlapping the segment $\{x(\ell)\}\times I_i$.
Since $I_i$ has a different parity than $x(\ell)$, this would make the algorithm report ``no tiling'' at line~\ref{alg:badoverlap}, contrary to our assumption.

We now prove by induction on $j$ that each $P_j$ can be tiled.
Since $P=P_k$, this is sufficient.
Along the way, we will also establish that $P_1\subset P_2\subset\ldots\subset P_k$.
When $j=1$, we see that $P_j$ is empty, so the statement is trivial.
Suppose now that $P_j$ can be tiled and consider $P_{j+1}$.
Note that if $x(e_j)=x(e_{j+1})$, so that $\ell$ does not move, then $P_j=P_{j+1}$, since all intervals that are created or modified when handling $e_{j+1}$ have the same parity as $x(\ell)$, so in this case, $P_{j+1}$ is tileable because $P_j$ is.

\begin{figure}[ht]
\centering
\includegraphics[page=2]{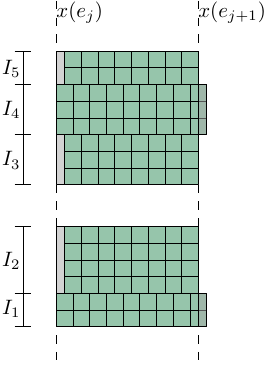}
\caption{The polyomino $P_j$ is the part of $P$ to the left of the line $x=x(e_j)$ (this part of $P_j$ is not shown) plus the grey rectangles along the line.
Here, the difference $x(e_{j+1})-x(e_j)$ is odd.
The difference $\overline{P_{j+1}\setminus P_j}$ has been tiled with green $2\times 2$ squares.}
\label{figure:tileddiff}
\end{figure}

Consider now the case $x(e_j)<x(e_{j+1})$.
Note that as $P_j\subset P$ and $P_j$ is to the left of the vertical line $x=x(e_j)+1$, we have $P_j\subset P_{j+1}$.
We now consider the set $\overline{P_{j+1}\setminus P_j}$ and argue that it is tileable; see Figure~\ref{figure:tileddiff}.
Let $I_1,\ldots,I_m$ be the intervals in $\mathcal I$ after $e_j$ was handled.
For each $I_i$, we add a rectangle $X\times I_i$ to $P_j$ in order to obtain $P_{j+1}$, where $X\subset \R$ is an interval with lower endpoint $x(e_j)$ or $x(e_j)+1$ and upper endpoint $x(e_{j+1})$ or $x(e_{j+1})+1$, and by the definition of $P_j$ and $P_{j+1}$, it follows that $X$ has even length.
Since each $I_i$ also has even length (otherwise, the algorithm would have returned ``no tiling'' at line~\ref{alg:oddlength} when $e_j$ was handled), the difference $\overline{P_{j+1}\setminus P_j}$ is a union of rectangles with even edge lengths, so $P_{j+1}$ is tileable since $P_j$ is.

\paragraph{Runtime analysis.}
Assuming that we can compare two coordinates in $O(1)$ time, we sort the vertical edges by their $x$-coordinates in $O(n\log n)$ time.
Since the intervals of $\mathcal I$ are pairwise interior-disjoint, we can implement $\mathcal I$ as a balanced binary search tree, where each leaf stores an interval~$I_i$.

We now argue that each vertical edge $e_j$, with $y$-coordinates $[y_0,y_1]$, takes only $O(\log n)$ time to handle, since we need to make only $O(1)$ updates to $\mathcal I$.
If the interior of $P$ is to the left of $e_j$, then $[y_0,y_1]\subset\bigcup_{i\in [m]} I_i$.
It then follows from the neighbor invariant that if $[y_0,y_1]$ overlaps more than one interval $I_i$, then the algorithm will return ``no tiling''.
We therefore do at most $O(1)$ updates to $\mathcal I$, so it takes $O(\log n)$ time to handle $e_j$.

On the other hand, if the interior of $P$ is to the right of $e_j$, we need to insert a new interval into $\mathcal I$ and possibly merge it with one or two neighbors in $\mathcal I$, so this also amounts to $O(1)$ changes to $\mathcal I$. 

At line~\ref{alg:ifmove}, we need to check the $O(1)$ intervals that were added or changed due to the edge $e_j$, so this can be done in $O(1)$ time.
Hence, the algorithm has runtime $O(n\log n)$.

\paragraph{Adaptation to $k\times k$ squares.}
In order to adapt the algorithm to $k\times k$ squares, we need to compare coordinates modulo $k$ instead of modulo $2$.
Specifically, each interval in $\mathcal I$ stores a number $p(I)\in\{0,1,\ldots,k-1\}$, which is set to $x(e_j)\text{ mod } k$ at line~\ref{alg:defP}.
We fail at line~\ref{alg:ifmove0} if $x(e_j)\text{ mod }k\neq p(I_i)$ and at line \ref{alg:ifmove} if some $I_i$ has a length not divisible by $k$.
At line~\ref{alg:merge}, we merge $I$ with the true neighbours that have the same $p$-value.
With these modifications, all arguments carry over to the case of $k\times k$ squares.

\section{Packing dominos}
In this section we will present our polynomial time algorithm for finding the maximum number of $1\times 2$ dominos that can be packed in a polyomino $P$. We assume that the dominos must be placed with axis parallel edges, but they can be rotated by $90^\circ$.
In any such packing, we can assume the pieces to have integral coordinates: if they do not, we can translate the pieces as far down and to the left as possible, and the corners will arrive at positions with integral coordinates.
We first describe a naive algorithm  which runs in polynomial time in the area of the polyomino.

\paragraph{Naive algorithm.} The naive algorithm considers the graph $G(P)=(V,E)$ where $V$ is the set of cells of $P$ and $e=(u,v)\in E$ if and only if the two cells $u$ and $v$ have a (geometrical) edge in common. The maximum number of $1\times 2$ dominos that can be packed in $P$ is exactly the size of a maximum matching of $G$ and it is well known that such a maximum matching can be found in polynomial time in $|V|$, i.e., in the area of $P$.
\\ \\
Our goal is to find an algorithm running in polynomial time, even with the compact representation of $P$ described in Section~\ref{sec:prelim}. In essence, we take the graph $G=G(P)$ above and construct from it a smaller graph, $G^*$, with $O(n^3)$ vertices. A maximum matching of $G^*$ yields an (implicit) description of a maximum matching of $G$ and we show that the maximum matching of $G^*$ can be found in time $O(\domcomp)$.

\subsection{Polynomial-time algorithm}\label{sec:algo}
We will next describe the steps of our algorithm for finding the maximum domino packing of a polyomino $P$. We first introduce the notion of a \emph{pipe} (see Figure~\ref{figure:pipedef}) and \emph{consistent parity}.
\begin{definition}\label{def:pipe}
Let $P$ and $Q$ be polyominos with $Q\subset P$. We say that $Q$ is a \emph{pipe} of $P$ if $Q$ is rectangular and both vertical edges of $Q$ or both horizontal edges of $Q$ are contained in edges of $P$. The \emph{width} of the pipe is the distance between this pair of edges. The \emph{length} of the pipe is the distance between the other pair of edges.
We say that a pipe is \emph{long} if its length is at least $3$ times its width.
\end{definition}

\begin{figure}[ht]
\centering
\includegraphics{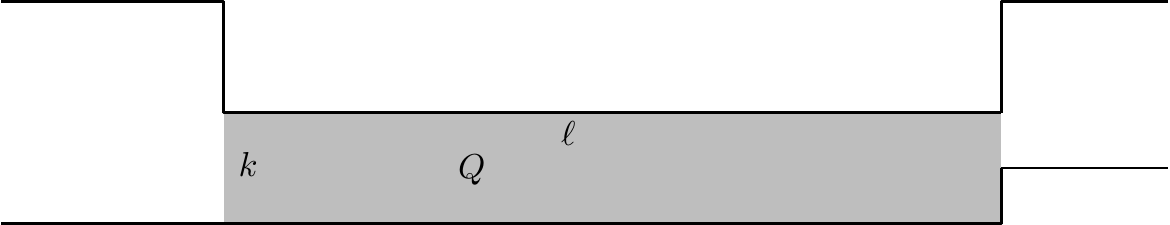}
\caption{A pipe of width $k$ and length $\ell$.}
\label{figure:pipedef}
\end{figure}

\begin{definition}\label{def:consistentparity}
We say that a polyomino $P$ has \emph{consistent parity} if all first coordinates of the corners of $P$ have the same parity and vice versa for the second coordinates.
Equivalently, $P$ has consistent parity if there exists an open $2\times 2$ square, $S$, such that for all choices of integers $i,j$ and $S'=S+(2i,2j)$, either $S'\subseteq P$ or $S'\cap P =\emptyset$.
\end{definition}


Next we present the steps of the algorithm. Figures~\ref{figure:algsteps12}--\ref{figure:algstep5} demonstrates the steps on a concrete polyomino $P$.

\begin{figure}[ht]
\centering
\includegraphics[page=4]{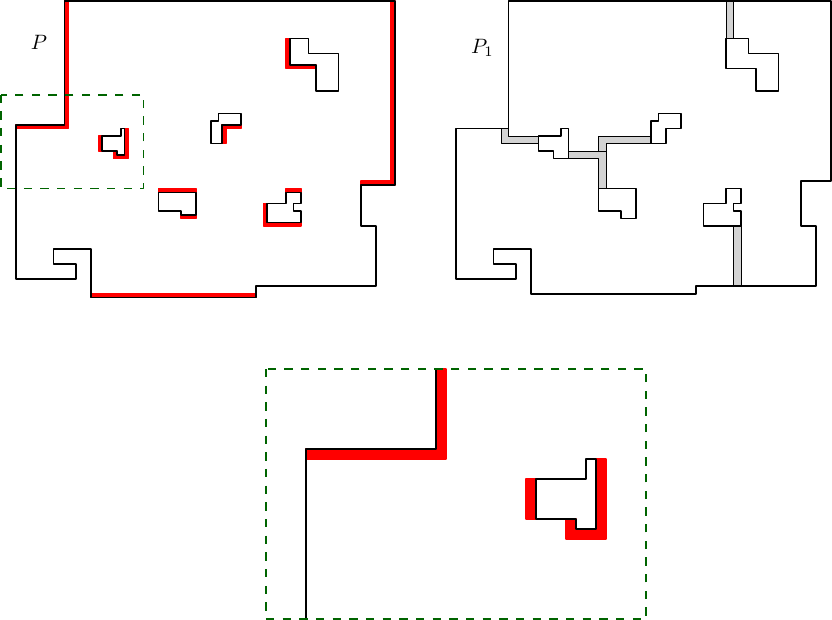}
\caption{Steps 1 and 2 of the algorithm.
Top left: Step 1, where the part $P\setminus P_1$, that is excluded from $P_1$ in order to make all coordinates even, is shown in red.
Top right: Step 2, where the holes of a polyomino $P_1$ are connected to the outer boundary by the grey channels.
Bottom: Closeup of the region in the dashed rectangle.}
\label{figure:algsteps12}
\end{figure}

\newcommand{\stepone}{Step 1: Compute the unique maximal polyomino $P_1\subset P$ with all coordinates even.}
\newcommand{\steptwo}{Step 2: Compute a polyomino $P_2\subset P_1$ with no holes and consistent parity by carving channels in $P_1$.}
\newcommand{\stepthree}{Step 3: Compute the offset $Q\mydef B(P_2,-\lfloor 3n/2 \rfloor)$ and then $P_3\mydef P\setminus Q$.}
\newcommand{\stepfour}{Step 4: Find the long pipes of $P_3$.}
\newcommand{\stepfive}{Step 5: Shorten the pipes and compute the associated graph $G^*$.}
\newcommand{\stepsix}{Step 6: Find the size of a maximum domino packing of $P$.}

\paragraph{\stepone}
We define $P_1$ to be the union of all $2\times 2$ squares $S$ of the form $S=[2i,2i+2]\times [2j,2j+2]$ with $i,j\in \Z$ and $S\subseteq P$. See the upper left and bottom part of Figure~\ref{figure:algsteps12}. 
It is readily checked that $P_1$ has at most $n$ corners. As we will see, $P_1$ can be computed in time $O(n\log n)$.

\paragraph{\steptwo}
Define $P'_0\mydef P_1$.
For $i=0,1,\ldots$, we do the following.
If there are holes in $P'_i$, we find a set of minimum size of $2\times 2$ squares $S_1,\ldots,S_k$ contained in $P'_i$ and with even corner coordinates that connects an edge of a hole to an edge of the outer boundary of $P'_i$.
To be precise, an edge of $S_1$ should be contained in the boundary of a hole of $P'_i$, an edge of $S_k$ should be contained in the outer boundary of $P'_i$, and for each $j\in \{1,\ldots,k-1\}$, $S_j$ and $S_{j+1}$ should share an edge.
We choose these squares such that they together form a $2\times 2k$ or $2k\times 2$ rectangle or an \textsf{L}-shape, which is clearly always possible.
We then define the polyomino $P'_{i+1}\mydef P'_i\setminus\bigcup_{j=1}^k S_j$, which has less holes than $P'_i$.
We stop when there are no more holes and define $P_2\mydef P'_i$ to be the resulting hole-free polyomino.
Note that in iterations $i\geq 1$, the holes may get connected to holes that were eliminated in earlier iterations or to channels carved in earlier iterations. See the upper right part of Figure~\ref{figure:algsteps12}. 
We will later see that $P_2$ has strictly less than $3n$ corners and that it can be computed in time $O(n^3)$.

\begin{figure}[ht]
\centering
\includegraphics[page=5]{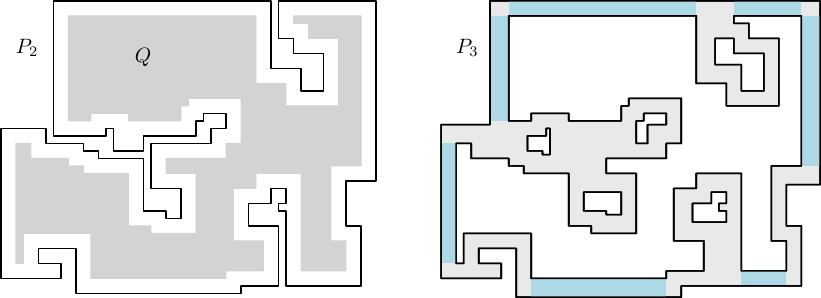}
\caption{Steps 3 and 4 of the algorithm, performed on the instance from Figure~\ref{figure:algsteps12}.
Left: Step 3, where the grey region $Q$ is an offset of the hole-free polyomino $P_2$.
In this example, $Q$ is connected, but that is in general not the case.
For pedagogical reasons, we offset by a smaller value than the algorithm would actually use.
Right: Step 4, where the grey and blue areas are $P_3\mydef P\setminus Q$.
The blue rectangles show the seven long pipes.}
\label{figure:algsteps34}
\end{figure}

\paragraph{\stepthree}
See the left part of Figure~\ref{figure:algsteps34}.
Note that we remove $Q$ from the original polyomino $P$ in order to get $P_3$, and not from $P_2$.
It is easy to check that $Q$ has at most $3n$ corners and consistent parity.
Hence $P_3\mydef P\setminus Q$ has at most $4n$ corners and, as we will see, $P_3$ has the property that for any $x\in P_3$, we have $\dist(x,\partial P_3)=O(n)$.
We will show how this step can be carried out in time $O(n \log n)$.

\paragraph{\stepfour}
Find all maximal long pipes $T_1,\dots,T_r$ in $P_3$ (recall that a pipe is long if its length is at least $3$ times its width). 
See the right part of Figure~\ref{figure:algsteps34}.
As we will see, there are at most $O(n)$ such pipes,  they are disjoint, and they each have width $O(n)$.
Later we will show how the pipes can be found in time $O(n \log n)$.

\paragraph{\stepfive}
Define $G_3\mydef G(P_3)$. We modify $G_3$ by performing the following shortening step for each $1\leq i \leq r$; see Figure~\ref{figure:algstep5}.
Assume with no loss of generality that the pipe $T_i$ is of the form $T_i=[0,\ell]\times [0,k]$ where $\ell$ is the length and $k\leq \ell/3$ is the width. If $\ell\leq 6$, we do nothing.
Otherwise, for each $j\in \{0,\dots,k-1\}$, we let  $S_j=[k+2,r] \times [j,j+1]$, where $r\mydef 2\lceil \ell/2 \rceil-k-2$, so that $G(S_j)$ is a horizontal path in $G_3$ consisting of an even number of vertices.
For each $j\in \{1,\dots,k-1\}$, we proceed by deleting the vertices of $S_j$ and their incident edges from $G_3$, and instead, we add an edge from the cell $[k+1,k+2]\times [j,j+1]$ to the cell $[r,r+1]\times [j,j+1]$ (i.e., we connect the cells to the left and right of $S_j$ with each other).

We denote the graph obtained after iterating over all $i$ by $G^*$.
Note that in $G^*$, there are only $O(k^2)=O(n^2)$ vertices corresponding to cells in each pipe $T_i$, since each pipe has width $k=O(n)$.
We show below that $G^*$ has $O(n^3)$ vertices and can be computed in time $O(n^3)$.
\begin{figure}[ht]
\centering
\includegraphics[page=2]{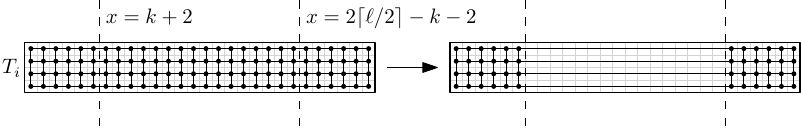}
\caption{Step 5 of the algorithm.
The part of the graph $G(T_i)$ in between the dashed vertical lines is substituted for long horizontal edges.}
\label{figure:algstep5}
\end{figure}


\paragraph{\stepsix}
We finally run a maximum matching algorithm on $G^*$.
Let $M$ be the resulting maximum matching, $N_0$ be the area of $P$, and $N_2$ be the number of vertices of $G^*$.
The algorithm outputs $|M|+(N_0-N_2)/2$ as the value of a maximum domino packing of $P$.
We show below that this step can be performed in time $O(\domcomp)$.
\\ \\
This completes the description of the algorithm. In Section~\ref{sec:structural}, we will provide some structural results on domino packings and polyominos. In Section~\ref{sec:correctness}, we will use these results to argue that the algorithm works correctly. 
In Section~\ref{sec:instancesize}, we will show that the reduced graph $G^*$ has $O(n^3)$ vertices and edges.
Finally, in Section~\ref{sec:runningtime}, we will use this to argue how the steps of the algorithm can be implemented with the claimed running times.

\subsection{Structural results on polyominos and domino packings}\label{sec:structural}
Building up to our structural results on domino packings, we require a definition and a few simple lemmas. 
\begin{figure}[ht]
\centering
\includegraphics{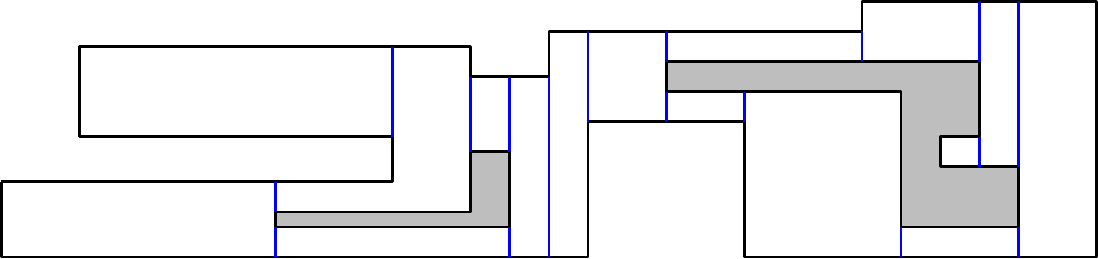}
\caption{A partition of a polyomino with two holes into rectangles using vertical line segments (blue).}
\label{figure:rectanglepartition2}
\end{figure}
Variations of the following lemma is well-known. We present a proof for completeness.
\begin{lemma}\label{lemma:rectanglepartition}
Let $P$ be an orthogonal polygon with $n$ corners and $h$ holes. $P$ can be divided into at most $n/2+h-1$ rectangular pieces by adding only vertical line segments to the interior of $P$. If $P$  is a polyomino, the rectangular pieces can be chosen to be polyominos too.
\end{lemma}

\begin{proof}
For each concave corner of the polygon we add a vertical line segment in the interior of the polygon starting from that corner and going upwards or downwards (depending on the rotation of the given corner). This is illustrated in~\Cref{figure:rectanglepartition2}. Let $s$ be the number of line segments added. It is easy to check that this gives a partition of $P$ into exactly $s-h+1$ rectangles. With $h$ holes, the number of concave corners is $n/2+2(h-1)$, so also $s\leq n/2+2(h-1)$ and the result follows.
\end{proof}
Note that for a polygon with $n$ corners, $h\leq (n-4)/4$, so we have the following trivial corollary.
\begin{corollary}\label{corollary:rectanglepartition}
The number of rectangular pieces in~\Cref{lemma:rectanglepartition} is at most $\frac{3}{4}n-2$.
\end{corollary}

We next show that the property of consisting parity is preserved under integral offsets.

\begin{lemma}\label{lemma:parityblowuop}
Let $P$ be a polyomino. If $P$ has consistent parity, then $B(P,1)$ and $B(P,-1)$ have consistent parity
\end{lemma}
\begin{proof}
Suppose $P$ has consistent parity. Let $S$ be a $2\times 2$ square as in~\cref{def:consistentparity}. Define $S_1=S+(1,1)$. It is easy to check that for all choices of integers $i,j$ and $S_1'\mydef S_1+(2i,2j)$, either $S_1'\subseteq B(P,1)$ or $S_1'\cap B(P,1) =\emptyset$. Thus $B(P,1)$ has consistent parity. The argument that $B(P,-1)$ has consistent parity is similar.
\end{proof}

\begin{lemma}\label{lemma:layerstructure}
Let $P$ be a connected polyomino of consistent parity and without holes. Define $L_1=B(P,1)\setminus P$ and $L_{-1}=P\setminus B(P,-1)$. Then $G(L_1)$ and $G(L_{-1})$ both have a Hamiltonian cycle of even length.
\end{lemma}
\begin{proof}
    To obtain a Hamiltonian cycle of $G(L_1)$, we can simply trace $P$ around the outside of its boundary, visiting all cells of $L_1$ in a cyclic order. The corresponding closed trail of $G(L_1)$ visits each vertex at least once. The assumption of consistent parity is easily seen to imply that we in fact visit each vertex exactly once, so the obtained trail is a Hamiltonian cycle.
    The graph $G(L_1)$ is bipartite, so the cycle has even length.
    The argument that $G(L_{-1})$ has a Hamiltonian cycle of even length is similar.
\end{proof}

With the above in hand, we are ready to state and prove our main structural results on domino packings. They are presented in~\cref{lemma:removingQ} and~\cref{lemma:pipes}.

\begin{lemma}\label{lemma:removingQ}
Let $P$ and $P_0$ be polyominos such that $P_0\subseteq P$, $P_0$ has no wholes, and $P_0$ has consistent parity. Let the total number of corners of $P$ and $P_0$ be $n$. Define $r=\lfloor \frac{3}{8}n \rfloor$ and $Q=B(P_0,-r)$.
There exists a maximum packing of $P$ with $1\times 2$ dominos which restricts to a tiling of $Q$.
\end{lemma}

Let us briefly pause to explain the importance of~\cref{lemma:removingQ}.
Suppose that $P$ contains a region $Q$ as described.
Then~\cref{lemma:removingQ} tells us that \emph{any} domino tiling of $Q$ can be extended to a maximum domino packing of $P$. We can thus disregard $Q$ and focus on finding a maximum packing of $P\setminus Q$, thus reducing the problem to a smaller instance.
This is one of our key tools for reducing the size of the original polyomino $P$ to a matching problem of polynomial size.
Another tool, namely to contract long pipes, will be described in Lemma~\ref{lemma:pipes} below, and in Section~\ref{sec:instancesize}, we will conclude that these two tools used carefully together reduce the packing problem to that of finding a maximum matching in a graph $G^*$ of size $O(n^3)$.

\begin{proof}
It follows from Lemma~\ref{lemma:parityblowuop} that $Q$ has consistent parity, and it can thus be tiled with $2\times 2$ squares and hence with dominos. Let $\mathcal{Q}$ be a tiling of $Q$. 

\begin{figure}[ht]
\centering
\includegraphics[page=2]{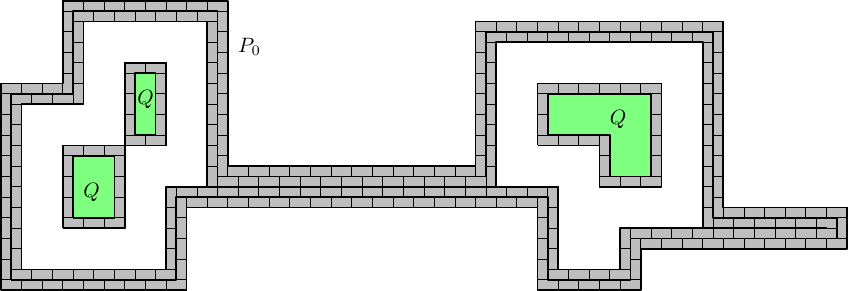}
\caption{The polyomino $P_0$ and the offset $Q$ (shown in green).
The figure also illustrates the 'layers' $A_i$ and their domino tilings, $\mathcal{A}_i$.}
\label{figure:peeling}
\end{figure}

Define $R=P\setminus P_0$ and note that $R$ has at most $n$ corners. It follows from~\Cref{corollary:rectanglepartition} that $R$ can be partitioned into less than $\frac{3}{4}n$ rectangular polyominos. Each of these rectangles has a domino packing with at most one uncovered cell (which happens when the total number of cells in the rectangle is odd). Fix such a packing $\mathcal{R}$ of the rectangles of $R$ with dominos. 

We next describe a tiling of $P_0\setminus Q$ as follows. For integers $1\leq i\leq r$ we define, $A_i=B(P_0,-i+1))\setminus B(P_0,-i)$. Intuitively, we can construct $Q$ from $P_0$ by peeling off the `layers' $A_i$ of $P_0$ one at a time. Let $i\in \{1,\dots,r\}$ be fixed. As $P_0$ has consistent parity, it follows from~\cref{lemma:parityblowuop} that $B(P_0,-i+1)$ has consistent parity.
It is also easy to check that $B(P_0,-i+1)$ has no holes either, and it then follows from~\cref{lemma:layerstructure} that each connected component of $G(A_i)$ has a Hamiltonian cycle of even length.
These cycles give rise to a natural tiling of $A_i$; if $(v_1,\dots,v_{2k})$ is the sequence of cells corresponding to such a cycle, then $\{v_1\cup v_2,v_3\cup v_4,\dots,v_{2k-1} \cup v_{2k}\}$ is a tiling of the cells of the cycle, and the union of such tilings over all connected components in $G(A_i)$ gives a tiling of $A_i$ with dominos.
Denote this tiling by $\mathcal{A}_i$.
See~\Cref{figure:peeling} for an illustration of this construction. 

Combining the tilings $\mathcal{A}_1,\dots,\mathcal{A}_r$ and $\mathcal{Q}$ with the packing $\mathcal{R}$, we obtain a domino packing, $\mathcal{P}$, of $P$ where at most $\frac{3}{4}n$ cells of $P$ are uncovered. We now wish to extend this packing to a maximum packing in a way where we do not alter the tiling $\mathcal{Q}$ of $Q$. If we can do this, the result will follow. Let $M$ be the matching corresponding to $\mathcal{P}$ in $G(P)$. We make the following claim. 

\begin{claim}
Let $k\leq r$. Suppose that the matching $M$ can be extended to a matching of size $|M|+k$. Then this extension can be made using a sequence $C_1,\ldots,C_k$ of $k$ augmenting paths one after the other (that is, $C_i$ is an augmenting path \emph{after} the matching has been extended using $C_1,\dots,C_{i-1}$) such that for each $i\in\{1,\ldots,k\}$, we have that $C_i$ only uses vertices of $G(R\cup \bigcup_{j=1}^i A_j)$.
\end{claim}

Before proving this claim, we first argue how the result follows.
Since there are less than $\frac{3}{4}n$ unmatched vertices in $M$, we can extend $M$ to a maximum matching using at most $r=\lfloor \frac{3}{8}n \rfloor$ augmenting paths. 
By the claim, these paths can be chosen so that they avoid the vertices of $G(Q)$.
In particular, we never alter the matching of $G(Q)$, so the final maximum matching restricted to $G(Q)$ is just the tiling $\mathcal Q$.
\begin{figure}[ht]
\centering
\includegraphics[page=3]{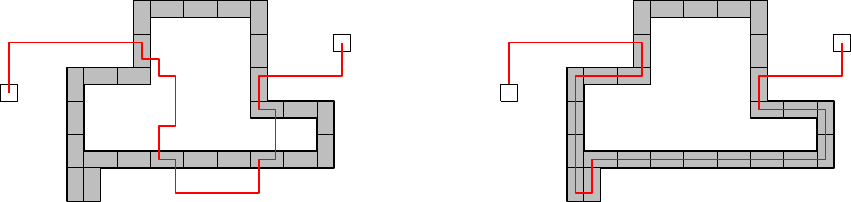}
\caption{Left: An alternating path between two unmatched vertices which enters a connected component of $G(A_k)$. Right: Modifying the alternating path using the the Hamiltonian cycle of the connected component.}
\label{figure:changingpath}
\end{figure}

We proceed to prove the claim by induction on $k$.
The statement is trivial for $k=0$, so let $1\leq k \leq r$ satisfy the assumptions of the claim and suppose inductively that $C_1,\ldots,C_{k-1}$ can be chosen such that for each $i\in\{1,\ldots,k-1\}$, we have that $C_i$ only uses vertices of $G(R\cup \bigcup_{j=1}^i A_j)$. After augmenting the matching using $C_1,\ldots,C_{k-1}$, we have only modified the matching restricted to $G(R\cup \bigcup_{j=1}^k A_j)$.  By Lemma~\ref{lemma:berge}, we can find an augmenting path $C_k'$ connecting two unmatched vertices $u,v$ of $G(P)$.
We will modify $C_k'$ to a path $C_k$ with $C_k \subset R  \cup   \bigcup_{j=1}^k A_j$.
Write $C_k':u=u_1,u_2,\dots,u_{2\ell}=v$.
Let $D$ be a Hamiltonian cycle of one of the connected components of $G(A_k)$; see Figure~\ref{figure:changingpath}.
If the path $C_k'$ ever enters the vertices of $D$, we let $i$ be minimal such that $u_i\in D$ and $j$ be maximal such that $u_j\in D$.
We can now replace the subpath $u_i,u_{i+1},\dots,u_j$ of $C_k$ with part of the Hamiltonian cycle $D$.
Whether we go clockwise or counterclockwise along $D$ depends on whether $u_i$ is matched with $u_{i+1}$ in a clockwise or counterclockwise fashion in $D$.
We do the same modification for every Hamiltonian cycle $D$ corresponding to a connected component of $G(A_k)$ that that $C_k'$ intersects.
Note that each cycle $D$ partitions the vertices $G(P)\setminus D$ into an interior and an exterior part.
Since $P_0$ has no holes and $u,v\in R$, the original path $C_k'$ enters $D$ from the exterior at $u_i$ and likewise leaves $D$ into the exterior at $u_j$.
Also note that $Q$ is contained in the interior parts of the cycles of $G(A_k)$.
It then follows that the final resulting path $C_k$ avoids $Q$ and $A_j$ for $j>k$, so it is contained in $R\cup \bigcup_{j=1}^k A_j$.

\end{proof}

\begin{figure}[ht]
\centering
\includegraphics{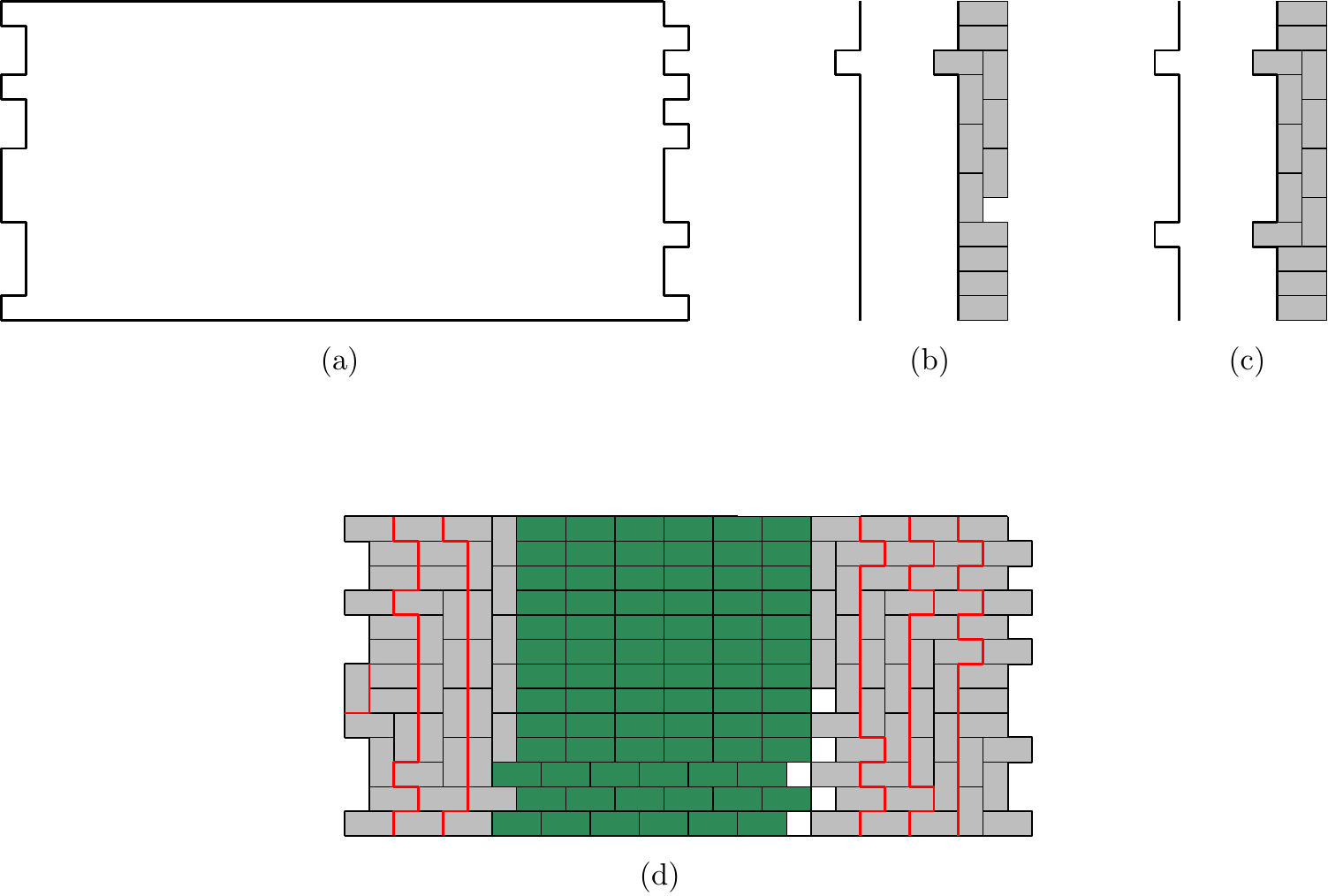}
\caption{(a) The polyomino $P$.
(b) Shifting a notch.
(c) Cancelling two notches.
(d) The partial packing obtained after shifting notches downwards and cancelling notches when possible (gray) and the horizontal dominos completing the packing (green).}
\label{figure:pipe}
\end{figure}

As it turns out, Lemma~\ref{lemma:removingQ} is not in itself sufficient to yield a polyomino with area $n
^{O(1)}$.
For example, $P$ may contain exponentially long and narrow pipes (see~\Cref{def:pipe}), say of width $n/10$, which will remain even when $Q$ is removed. Surprisingly, it turns out that such narrow pipes are the only obstacles that prevent us from reducing to an instance of size polynomial in $n$. This is what motivates the following lemma which intuitively yields a reduction for shortening long narrow pipes.

\begin{lemma}\label{lemma:pipes}
Let $k,\ell \in \N$ with $\ell$ even. Let $L\subseteq [-1,0]\times [0,k]$, $R\subseteq [\ell,\ell+1]\times [0,k]$ be polyominos and  define $P=L\cup R\cup ([0,\ell]\times[0,k])$. Color the cells of the plane in a chessboard like fashion and let $b$ and $w$ be respectively the number of black and white cells contained in $P$. Assume without loss of generality that $b\geq w$.
If $\ell\geq 2k$, then the number of uncovered cells in a maximum domino packing of $P$ is exactly $b-w$.
Moreover, there exists a maximum domino packing such that the rectangle $[k+1,\ell-k-1]\times [0,k]$ is completely covered and all dominos intersecting the rectangle are horizontal.
\end{lemma}

\begin{proof}
As each domino covers one black and one white cell, any packing will leave at least $b-w$ cells uncovered. We thus need to demonstrate the existence of a packing with exactly $b-w$ uncovered cells. To see that such a packing exists, it is very illustrative to consider~\cref{figure:pipe}. An example of a polyomino, $P$, is illustrated in ~\cref{figure:pipe}(a). We first tile as many cells of $L$ and $R$ as possible, such that no two uncovered cells of $L$ and $R$ share an edge. We will call these uncovered cells \emph{notches}. We next show how we can alter the configuration of notches by only adding a layer of width $2$. First, we note that a notch can be shifted an even number of cells downwards or upwards using the construction in
~\cref{figure:pipe} (b). In case we have two notches of different colours in the chessboard coloring and with no other notches between them, we can use the construction in~\cref{figure:pipe} (c) to cancel  these two notches from the configuration of notches. Our goal is to use the constructions of (b) and (c) to shift the notches of $L$ and $R$ downwards, cancelling notches if possible. Going through the notches of $L$ from bottom to top, we shift them down as far as possible using the construction in (b). In case a notch has a different color than the nearest notch below it, we use construction (c) to cancel them. We further add horizontal dominos such that the configuration of notches is preserved at all other positions than where the shifting or cancelling occurs. We do a similar thing for $R$. The process from start to end is illustrated in~\cref{figure:pipe}(d) which also shows the resulting partial tiling. The red lines in the figure separate the steps of the process.

Note that each added layer of the process has thickness $2$. Initially, each of $L$ and $R$ consists of at most $\lceil k/2\rceil$ notches (after the first step which isolates the notches). Moreover, if $k$ is odd and there is $\lceil k/2\rceil$ notches in $L$ or $R$, then no shifting/cancelling is needed on that side.
It then follows from the assumption $l\geq 2k$ that we are able to finish this partial packing. Let $b'$ and $w'$ be the number of black and white cells uncovered by this partial packing. Then $b'-w'=b-w$. It is also easy to check that we can complete the packing of $P$ using only horizontal dominos and leaving exactly $b'-w'=b-w$ cells uncovered. This completes the proof.
\end{proof}


Lemma~\ref{lemma:pipes} allows us to 'shorten' long and narrow pipes when searching for the maximum domino packing. It follows from the Lemma that going from $G_3$ to $G^*$ in the shortening step 5 of our algorithm does not alter the number of unmatched vertices in a maximum matching. We will return to this in Section~\ref{sec:correctness}.
Note that, unlike in~\cref{lemma:removingQ}, we cannot simply remove (part of) the pipe from the polyomino. 
The following lemma allows us to upper bound the size of a set of overlapping pipes no two of which are contained in the same larger pipe.


\begin{lemma}\label{lemma:planargraph}
Let $G=(V,E)$ be a graph of order $n\geq 2 $ with no self-loops but potential multiple edges. Suppose that $G$ has a planar embedding such that for any pair of multiple edges $(e_1,e_2)$, the Jordan curve formed by $e_1$ and $e_2$ in the planar embedding of $G$ contains a vertex of $G$ in its interior. Then the number of edges of $G$ is upper bounded by $3n-5$.
\end{lemma}
\begin{proof}
In what follows, we will use the classic result that the  number of edges of a simple planar graph of order $n$ is upper bounded by $3n-6$.  
We prove the result by strong induction on $n$. For $n=2$ the result is trivial so let $n>2$ be given and suppose the bound holds for smaller values of $n$. 
Let $\mathcal{E}$ be a planar embedding of $G$. Let $(e_1,e_2)$ be a pair of distinct multiple edges that is minimal in the sense that no other such pair $(e_1',e_2')$ exists with the following property: If $\gamma$ and $\gamma'$ are the Jordan curves formed by $(e_1,e_2)$ and  $(e_1',e_2')$ in $\mathcal{E}$, then $\Int \gamma ' \subset  \Int \gamma $. 
Assume that $e_1$ and $e_2$ connect vertices $u$ and $v$.
Let $V'$ be the set of vertices of $G$ that are contained in $\Int \gamma$ under $\mathcal{E}$ and let $k=|V'|$. 
Then, $1 \leq k \leq n-2$.
Let $G_1=(V_1,E_1)$ where $V_1=V'\cup\{u,v\}$ and $E_1$ is formed by $e_1$  together with all edges of $G$ that are incident to a vertex in $V'$. Let $G_2=(V_2,E_2)$ where $V_2=V\setminus V'$ and $E_2=E\setminus E_1$. 
Clearly $G_1$ is a simple planar graph on $k+2$ vertices. Moreover, it is readily checked that $G_2$ is a planar graph on $n-k$ vertices which satisfies the assumptions of the lemma. Note that $2\leq n-k<n$. 
It thus follows from the inductive hypothesis that the number of edges of $G$ is upper bounded by
$$
3(n-k)-5+3(k+2)-6=3n-5.
$$
This completes the proof.


\end{proof}

\begin{lemma}\label{lemma:numberofpipes}
Let $P$ be a polyomino with $n$ corners. Let $Q_1,\dots,Q_r$ be pairwise disjoint pipes of $P$, no two of which are contained in a larger pipe of $P$, i.e., for no two distinct $i,j$ does there exist a pipe $Q$ with $Q_i\cup Q_j\subseteq Q$. Then $r\leq 3n-5$.
\end{lemma}
\begin{proof}
We construct a graph $G=(V,E)$ as follows. $V$ is the set of (geometric) edges of $P$. For each $i\in\{1,\dots,r\}$ we let $u_i,v_i\in V$ be the two parallel edges of $P$ which contain two opposite sides of $Q_i$ and we add the edge $(u_i,v_i)$ to $E$. $G$ is thus a graph of order $n$ with exactly $r$ edges. We note that $G$ may have multiple edges but it has a natural planar embedding, $\mathcal{E}$, such that for each pair of multiple edges $(e_1,e_2)$ the Jordan curve formed by $e_1$ and $e_2$ under $\mathcal{E}$ contains a vertex of $G$ in its interior (see Figure~\ref{fig:pipegraph}). Here we used that for any two $i,j$ with $1\leq i <j\leq r$, the two pipes $Q_i$ and $Q_j$ are not contained in a larger pipe of $P$. It thus follows from Lemma~\ref{lemma:planargraph} that $r\leq 3n-5$.

\end{proof}
\begin{figure}[ht]
\centering
\includegraphics{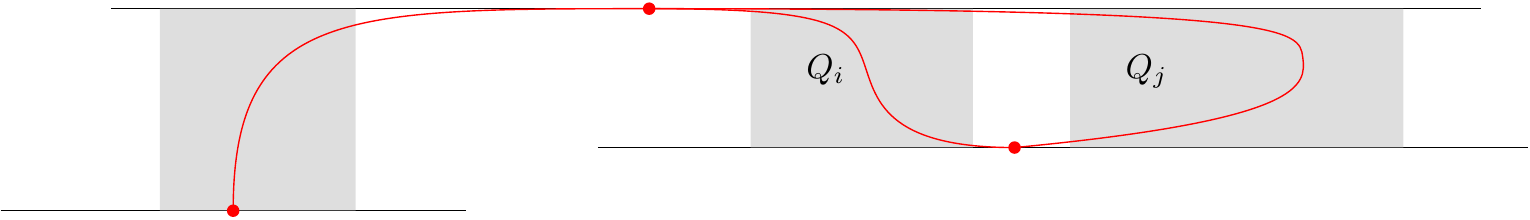}
\caption{The construction of the planar graph $G$. Pipes are shown in grey. The white rectangle between, $Q_i$ and $Q_j$ must contain some nonempty subset of $\partial P$ in its interior --- otherwise this rectangle could joined with $Q_i$ and $Q_j$ forming a larger pipe of $P$ containing both $Q_i$ and $Q_j$.}
\label{fig:pipegraph}
\end{figure}

We need to argue that the algorithm outputs the correct value and that the different steps can be implemented to obtain the stated running times.

\subsection{Correctness of the domino packing algorithm}\label{sec:correctness}
We now refer the reader back to the description of our domino packing algorithm from Section~\ref{sec:algo} and show that it correctly finds the size of a maximum domino packing. 
To show this, it suffices to show that maximum matchings of $G_0$ and $G^*$, leave the same number of unmatched vertices. First note that $P_1$ has at most $n$ corners. Further, a polyomino with $n$ corners can have at most $(n-4)/4$ holes, and since we remove a hole and add at most $6$ new corners in going from $P_i'$ to $P_{i+1}'$,  $P_2$ has at most $5n/2<3n$ corners. It follows that also $Q=B(P_2,\lfloor 3n/2 \rfloor)$ has at most $3n$ corners. Letting $n_1$ denote the number of corners of $P_3=P\setminus Q$ it finally follows that $n_1\leq 4n$. Now $\dist(\partial P,Q)\geq \lfloor 3n/2 \rfloor\geq \lfloor \frac{3}{8}n_1 \rfloor$ and moreover, $Q$ has consistent parity and no holes, so Lemma~\ref{lemma:removingQ} applies, giving that $G_0$ has a maximum matching, which restricts to a perfect matching of $G(Q)$ and to a maximum matching of $G_3$. In particular, maximum matchings of $G_0$ and $G_1$ leave the same number of vertices unmatched.

Next, we argue that maximum matchings of $G_3$ and $G^*$ again leave the same number of vertices unmatched. It is easy to see that a maximum matching of $G^*$ can be extended to a matching of $G_3$ with the same number of unmatched vertices by simply inserting more horizontal dominos in the horizontal pipes and vertical dominos in the vertical pipes (here we use that the $S_j$'s as defined in Step 3, each consists of an even number of cells).
Conversely, let $M_1$ be a maximum matching of $G_3$. We show that $G^*$ has a matching, $M_2$, with the same number of uncovered cells. 
For this we consider the pipes $(T_i)_{i=1}^r$ found in step 4 of the algorithm. For each $1\leq i \leq r$, we let $T_i'$ be the pipe obtained from $T_i$ by shortening $T_i$ by one layer of cells in each end. The length of $T_i'$ is thus two shorter than that of $T_i$. Let further $L_i\supseteq T_i'$ consist of all cells of $P$ which are covered by a domino which cover at least one cell of $T_i'$. The sets $(L_i)_{i=1}^r$ are pairwise disjoint and they are each of the form of the set $L$ in Lemma~\ref{lemma:pipes} (up to a 90 degree rotation). Moreover, the maximum matching $M_1$ restricts to a maximum matchings of  $M_1'$ of $G(P_3\setminus \bigcup_{i=1}^r L_i)$ and a maximum matching $M_1
^{(i)}$ of $G(L_i)$ for $1\leq i \leq r$. For $1\leq i \leq r$, we let $G_3
^{(i)}=G(L_i)$ and $G_2^{(i)}$ be the corresponding subgraph of $G_2$. We define $M_2$ to be $M_1'$ combined with any maximum matchings of the $G_2^{(i)}$, $1\leq i \leq r$. By applying Lemma~\ref{lemma:pipes} to each $L_i$, it follows that the maximum matchings of $G_3^{(i)}$ and $G_2^{(i)}$ leave the same number of unmatched vertices. It thus follows that $M_2$ and $M_1$ leave the same number of unmatched vertices. This finishes the argument that the algorithm works corrrectly.


\subsection{Bounding the size of the reduced instance}\label{sec:instancesize}
In determining the running time of our algorithm, it is crucial to bound the size of the reduced instance $G^*$.
In this section we show that $G^*$ has $O(n^3)$ vertices.
As explained in the next section, we can then find a maximum matching of $G^*$ in $O(n^3 \,\text{polylog}\, n )$ time.

We start out by proving the following lemma.
\begin{lemma}\label{lem:distbound}
The polyomino $P_3$ contains no $63n \times 63 n$ square subpolyomino.
\end{lemma}
\begin{proof}
Let $n'=\lfloor 3n/2 \rfloor$. 
We show that $P_3$ contains no $41n'\times 41n'$ square as a subpolyomino and the desired result will follow. Suppose for contradiction that $S\subseteq P_3$ is such a subpolyomino.
Note that $Q$ consists of exactly those points of $P_1$ of distance at least $n'$ to all the channels of $C:=P_1\setminus P_2$ and to $\partial P_1$.
Thus any point $x\in P_3$ has distance at most $n'$ to $C$ or to $\partial P_1$. In particular $S$ contains a $39n' \times 39n'$ square subpolyomino, $S_1\subseteq P_1$, all points of which are of distance at least $n'$ to $\partial P_1$ and thus, of distance at most $n'$ to $C$. 

\begin{figure}[ht]
\centering
\includegraphics{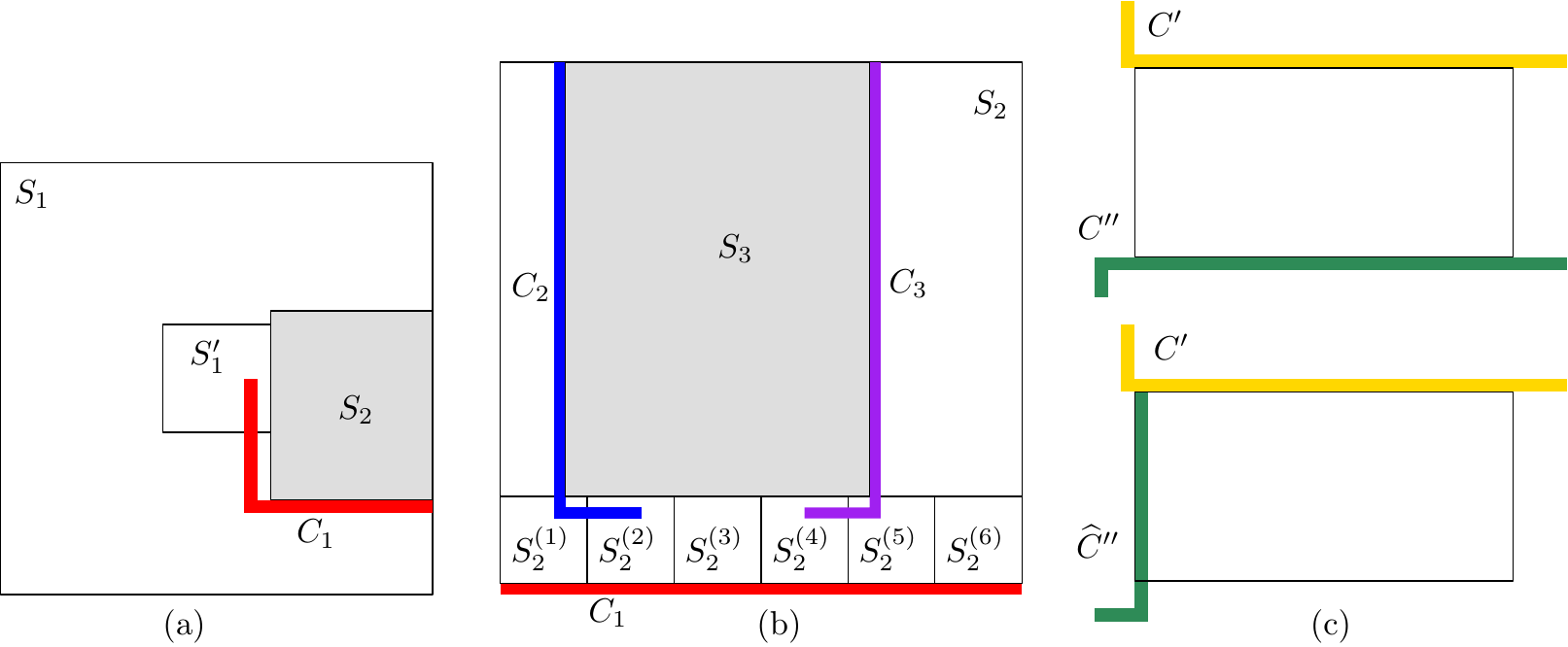}
\caption{Situations in the proof of Lemma~\ref{lem:distbound}.}
\label{fig:parallelchannels}
\end{figure}

By the way we chose the channels, each channel connects a hole of $P_1$ with either the boundary of $P_1$ or with a channel already carved in an earlier iteration.
Since $\partial P_1 \cap S_1 =\emptyset$, it follows that any channel intersecting $S_1$ has an end outside $S_1$ and thus leaves $S_1$ through an edge of $S_1$. 

Write $S_1'$ for the central $3n' \times 3n'$ square polyomino of $S_1$; see Figure~\ref{fig:parallelchannels} (a).
Since any point of $S_1$ is of distance at most $n'$ to $C$, $S_1'$ must intersect a channel $C_1\subset C$ (depicted in red in the figure).
We know that $C_1$ leaves $S_1$. 
It is simple to check that this leads to the existence of an $18n' \times 19n'$ rectangular polyomino $S_2\subseteq S_1$ having along one of its sides a straight part of the channel $C_1$ of length $18n'$; see Figure~\ref{fig:parallelchannels} (b).
Assume with no loss of generality that $S_2=[0,18n']\times [0,19n']$ and that the channel $C_1$ runs along the base of the rectangle $S_2$ as in the figure.
For $1\leq i \leq 6$ we define $S_2^{(i)}$ to be the square polyomino $[3(i-1)n',3in']\times [0,3n']$.
By the same reasoning as above, each of these squares must intersect the set of channels $C$ non-trivially.
As each channel turns at most once by construction, the squares $S_2^{(i)}$ are disjoint from $C_1$.
To finish the proof, we require the following claim.
\begin{claim}
Let $B=[0,k] \times [0,\ell]$, $k,\ell\in \N$ be a $k\times \ell$ square polyomino.
Suppose that $[0,k] \times [-1,0]$ is contained in some channel, $C'$, and that $[0,k] \times [\ell,\ell+1]$ is contained in some other channel, $C''$.
Then $k\leq \ell+2$.
\end{claim}
\begin{proof}[Proof of Claim.]
See Figure~\ref{fig:parallelchannels} (c).
Suppose without loss of generality that $C'$ was carved in iteration $i$ and $C''$ was carved in iteration $j$ in the process of generating $P_2$ in step 2 of the algorithm, and that $i<j$.
The channel $C''$ was chosen to connect a yet unconnected hole $H$ of $P_{j-1}'$ with the outer boundary of $P_{j-1}'$ along a shortest path in the $L_\infty$-norm.
At the time $C''$ was carved, $C'$ had already been carved and thus the edges of $C'$ (except the two ``ends'' of length $2$) is part of the outer face of $P_{j-1}'$.
Under the assumption $k> \ell+2$, we can find a shorter path connecting $H$ to $\partial P_{j-1}'$; see the bottom part of Figure~\ref{fig:parallelchannels} (c).
This shorter path shows that we could have picked a shorter channel $\widehat C''$ in place of $C''$.
This is a contradiction, so we conclude that $k\leq \ell+2$.
\end{proof}

Let us now finish the proof of the lemma.
We know that the two squares $S_2^{(2)}$ and $S_2^{(5)}$ each intersect channels of $C$.
Let us denote these not necessarily distinct channels respectively $C_2$ and $C_3$.
If these channels are the same, the channel $C_2$ passes straight through $S_2^{(3)}$.
But then the two channels $C_1$ and $C_2$ run in parallel for a length of at least $3n'+4$ and they have distance at most $3n'$ which gives a contradiction with the claim.
Thus $C_2$ and $C_3$ are different channels.
By the same reasoning as for $C_1$, the channel $C_2$ must leave $S_2$. 
If it does so in a direction parallel to $C_1$, we similarly obtain a contradiction with the claim.
Thus, it most leave $S_2$ in a direction perpendicular to $C_1$.
The same logic applies to $C_3$; see Figure~\ref{fig:parallelchannels} (b).
Now the two channels $C_2$ and $C_3$ provide a contradiction to the claim.
Indeed, their straight segments span a box, $S_3$, of dimensions $\ell \times  18n'$ where $\ell\leq 18n'-4$.
With this contradiction, we conclude that $P_3$ contains no $41n'\times 41n'$ square as a subpolyomino and the proof is complete.
\end{proof}
\begin{remark}
    No serious effort has been made to optimize the constants in Lemma~\ref{lem:distbound}.
\end{remark}

\begin{corollary}\label{cor:distbound}
For any $x\in P_3$, we have $dist(x,\partial P_3)\leq 32n$.
\end{corollary}
\begin{proof}
If not, $P_3$ contains a $63n \times 63 n$ square, a contradiction.
\end{proof}
Corollary~\ref{cor:distbound} shows that each point of $P_3$ is of distance $O(n)$ to the boundary of $P_3$. In particular, this shows that the long pipes, $T_1,\dots,T_r$, found in step 4 of the algorithm all have width at most $O(n)$. By Lemma~\ref{lemma:numberofpipes}, $r=O(n)$, so when performing the shortening reduction in step 5 of our algorithm, the part of $G^*$ contained in contracted pipes 
gets size $O(n^3)$. We finish this section by showing that $P_3\setminus \bigcup_{i=1}^r T_i$ also consists of $O(n^3)$ cells. From this it will follow that $G^*$ is of order $O(n^3)$ which is what we require. We state the result as a lemma.
\begin{lemma}\label{lem:finalSize}
The reduced instance $G^*$ found by our algorithm has $O(n^3)$ vertices and edges.
\end{lemma}
\begin{proof}
For technical reasons to be made clear shortly we define $T_i'$ to be the pipe obtained from $T_i$ by shortening $T_i$ by a layer of cells in each end. The length of  $T_i'$ is thus exactly the length of $T_i$ minus 2. 
Let $R$ be the polyomino $P_3\setminus \bigcup_{i=1}^r T_i'$. Consider a (geometrical) edge, $e$, in the set $\partial R \setminus \partial P_3$ which is an edge forming an end of a shortened pipe $T_i'$.
It then follows that the endpoints of $e$ are corners of $R$ and in particular that $e$ is not contained in a longer edge of $\partial R$ (this would not necessarily be the case if we had not shortened the pipes a layer in each end when defining $R$). We will use this observation shortly.

As discussed, it suffices to show that $R$ has $O(n^3)$ cells.  Note that $R$ has $O(n)$ corners:
Indeed, $R$ is obtained from $P_3$, which has $O(n)$ corners, by removing $O(n)$ pipes, each of which adds only $4$ corners.
We show that any point  $x\in R$ is of distance $O(n)$ from a corner. Since each corner can have at most $O(n^2)$ cells within distance $O(n)$, this will show that $R$ has $O(n^3)$ cells. 

So let $x \in R$ be arbitrary. Also let $c\mydef 32$. By Corollary~\ref{cor:distbound}, any point of $P_3$ is of distance at most $cn$ to $\partial P_3$. It follows that, similarly, any point of $R$ is of distance at most $cn$ to $\partial R$. Let $S$ be the $6cn \times 6cn$ square centered at $x$ and suppose that $S$ contains no corner of $R$. Then $\partial R\cap S$ is a collection of horizontal and vertical straight line segments. Moreover, they are either all horizontal or all vertical as otherwise, they would intersect in a corner of $R$ inside $S$.
Assume without loss of generality that they are all horizontal.
Using that any point of $R$ is of distance at most $cn$ to $\partial R$, it follows that there exists two such parallel segments of distance at most $2cn$, one being above $x$ and one being below. Take a closest pair of such segments. Together they form a pipe, $T$, of $R$ of length $6cn$ and width at most $2cn$, i.e., a pipe of a length at least three times its width. Now $T$ is disjoint from the pipes $T_1,\dots, T_r$ (since $T\subset R$ and each $T'_i$ is disjoint from $R$). $T$ is a pipe of $R$ but we claim that it is in fact also a pipe of $P_3$. To see this, we note that by Corollary~\ref{cor:distbound}, the pipes found in step 4 of our algorithm have width at most $2cn$. In particular, the edges of $\partial R\setminus \partial P_3$ have length at most $2cn$ and we saw that they are not contained in longer edges of $\partial P_3$. However, the pipe $T$ has length $6cn$ and so, the long edges of $T$ are in fact edges of $\partial P_3$, so $T$ is a pipe of $P_3$.
This contradicts the maximality of the set of pipes $T_1,\dots,T_r$. We thus conclude that $S$ must contain a corner of $\partial R$. Since $x \in R$ was arbitrary, this shows that any point in $R$ is of distance at most $3cn=O(n)$ to a corner of $R$ and the proof is complete.
\end{proof}
\subsection{Implementation of the individual steps}\label{sec:runningtime}
We next describe how the different step of our domino tiling algorithm can be implemented.
\paragraph{\stepone}
We first compute the set $P_1\subset P$ with consistent parity.
To obtain $P_1$, we move all corners of $P$ to the interior of $P$ to the closest points with even coordinates as shown in Figure~\ref{figure:algsteps12}.

Moving the corners may cause some corridors of $P$ to collapse (namely the corridors of $P$ of thickness $1$), so that $P_1$ has overlapping edges corresponding to degenerate corridors.
The degenerate vertical corridors can be filtered out in $O(n\log n)$ time as follows (and the degenerate horizontal ones are handled analogously).
We sort the vertical edges after their $x$-coordinates and thus partition the edges into groups with identical $x$-coordinates.
For each group of vertical edges with the same $x$-coordinate, we sort them according to the $y$-coordinates of their lower endpoints.
Let $e_1,\ldots,e_k$ be one such sorted group.
We run through the edges $e_1,\ldots,e_k$ in this order.
For each edge $e_i$, we run through the succeeding edges until we get to an edge $e_j$ which is completely above $e_i$.
For each of the overlapping edges $e_k$, $k\in\{i+1,\ldots,j-1\}$, we remove the corresponding degenerate corridor of $P_1$ created by $e_i$ and $e_k$.
Since no triple of edges can be pairwise overlapping, there are $O(n)$ overlapping pairs in total, so this process is dominated by sorting, which takes $O(n\log n)$ time.

\paragraph{\steptwo}
Remember that we need to find a set of minimum size of $2\times 2$ squares $S_1,\ldots,S_k$ contained in $P'_i$ and with even coordinates that connects an edge of a hole to an edge of the outer boundary of $P'_i$.
For each pair of an edge of a hole of $P'_i$ and an edge of the outer boundary of $P'_i$, we can compute the size of the smallest set of squares connecting those two edges in $O(1)$ time, so by checking all pairs, we find the overall smallest set in $O(n^2)$ time.
Note that no edge of a middle square $S_j$, $2\leq j\leq k-1$, is contained in the boundary $\partial P'_i$, since otherwise, there would be a smaller set of squares with the desired properties.
Therefore, constructing $P'_{i+1}\mydef P'_i\setminus\bigcup_{j=1}^k S_j$ can then be done in $O(1)$ time once the squares $S_j$ have been found.
Since there are initially $O(n)$ holes to eliminate, the process takes $O(n^3)$ time in total.

\begin{figure}[ht]
\centering
\includegraphics[page=3]{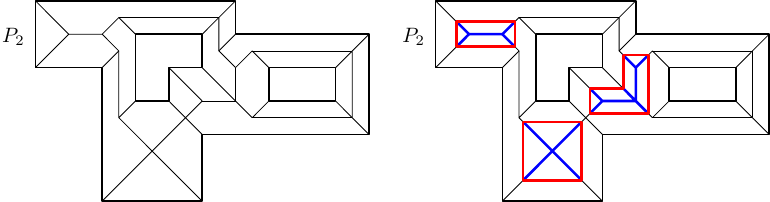}
\caption{Left: The $L_\infty$ Voronoi diagram $VD=VD(P_2)$ of the edges of a polyomino $P_2$.
Right: The blue parts of $VD$ are the subgraphs that have some sufficient distance $d$ to the boundary $\partial P_2$.
The red cycles enclose the regions of $P_2$ with at least distance $d$ to $\partial P_2$.}
\label{figure:voronoi}
\end{figure}

\paragraph{\stepthree}
We now explain how to compute the set $Q\subset P_2$ defined by $Q\mydef B(P_2,-\lceil 3n/4\rceil)$.
The boundary of $Q$ can be computed from the $L_\infty$ Voronoi diagram $VD\mydef VD(P_2)$ of the edges of $P_2$ by a well-known technique described by
Held, Luk{\'a}cs, \& Andor~\cite{held1994pocket}, as follows.
The Voronoi diagram $VD$ is a plane graph contained in $P_2$ that partitions $P_2$ into one region $R(e_i)$ for each edge $e_i$ of $P_2$, such that if $x\in R(e_i)$ then $\dist(x,e_i)=\dist(x,\partial P_2)$, and $VD$ consists of horizontal and vertical line segments and line segments that make $45^\circ$ angles with the $x$-axis; see Figure~\ref{figure:algstep5} (left).

We find all maximal subgraphs of $VD$ consisting of points with distance at least $\lceil3n/4 \rceil$ to $\partial P_2$.
The leafs of each subgraph $G$ lie on a cycle in $P_2$ where the distance to $\partial P_2$ is constantly $\lceil3n/4 \rceil$, and the cycles can be found by traversing the leafs of $G$ in clockwise order; see~\cite{held1994pocket} for the details and Figure~\ref{figure:algstep5} (right) for a demonstration.

We can compute $VD$ in $O(n\log n)$ time using the sweep-line algorithm of Papadopoulou \& Lee~\cite{papadopoulou2001voronoi}.
In our special case where all edges are horizontal or vertical, the algorithm becomes particularly simple as described by
Mart{\'\i}nez, Vigo, Pla-Garc{\'\i}a, \& Ayala~\cite{martinez2010skeleton}.
Once we have $VD$, it takes $O(n)$ time to compute $\partial Q$.

One can avoid the computation of $VD$ by offsetting the boundary into the interior by distance $1$ repeatedly $\lceil3n/4 \rceil$ times.
After each offset, we remove collapsed corridors as described in step 1.
This would take in total $O(n^2\log n)$ time.

The representation of $P_3\mydef P\setminus Q$ is obtained by simply adding the cycles representing the boundary of $Q$ to the representation of $P$.

\paragraph{\stepfour}
Recall that each long pipe $T_i\subset P_3$ is a maximal rectangle with a pair of edges contained in $\partial P_3$ which are at least $3$ times as long as the other pair of edges.

To find the pipes, we compute the $L_\infty$ Voronoi diagram $VD_1\mydef VD(P_3)$ of the edges of $P_3$ in $O(n\log n)$ time, as described in step 3.
Consider a long pipe $T_i$.
Assume without loss of generality that $T_i=[0,\ell]\times [0,k]$ where $\ell$ is the length and $k\leq \ell/3$ is the width.
We now observe that the segment $e\mydef [k/2,\ell-k/2]\times k/2$ in the horizontal symmetry axis of $T_i$ is contained in an edge of $VD_1$, since for any point $p$ in $e$, the edges of $P_3$ closest to $p$ are the horizontal edges of $T_i$.

Each horizontal or vertical edge $e$ of $VD_1$ separates the regions of points that are closest to a pair $s_1,s_2$ of horizontal or vertical edges of $P_3$.
It is easy to check whether $s_1,s_2$ define a long pipe containing (a part of) $e$.
Hence, all pipes can be identified by traversing the edges of $VD_1$.
As $VD_1$ has complexity $O(n)$, this step takes $O(n\log n)$ time in total.

\paragraph{\stepfive}
Recall that we define $G_3\mydef G(P_3)$ and obtain the final graph $G^*$ by replacing long horizontal (resp.~vertical) paths in pipes with long horizontal (resp.~vertical) edges.
Once $P_3$ and the long pipes have been computed, it is straightforward to construct $G^*$ in $O(n^3)$ time, since the size of $G^*$ is $O(n^3)$ by Lemma~\ref{lemma:rectanglepartition}.

\paragraph{\stepsix}
Recall that our algorithm outputs $|M|+(N_0-N_2)/2$, where $M$ is a maximum matching of $G^*$, $N_0$ is the area of $P$, and $N_2$ is the number of vertices of $G^*$.
In order to compute $M$, we use the multiple-source multiple-sink maximum flow algorithm by Borradaile, Klein, Mozes, Nussbaum, \& Wulff{-}Nilsen~\cite{borradaile2017multiple}.
In a directed plane graph with $m$ vertices and edge capacities, where a subset of the vertices are \emph{sources} and another subset are \emph{sinks}, the algorithm finds the maximum flow from the sources to the sinks respecting the capacities in time $O(m\log^3 m)$.
Gawrychowski \& Karczmarz~\cite{gawrychowski2018improved} described an improved algorithm with running time $O(m\frac{\log^3 m}{\log^2\log m})$.
Our graph $G^*$ is bipartite, so a maximum matching equals a maximum flow between the two vertex classes, when each edge has capacity $1$.
We can therefore find the maximum matching of $G^*$ in time $O(\domcomp)$.

\begin{remark}
We note that we can also find an (implicit) description of a maximum domino packing of $P$.
We simply extend the matching of $G^*$ by inserting horizontal dominos in the horizontal pipes and vertical dominos in the vertical pipes.
We further decide to give $Q$ any standard tiling, e.g., the one that uses only horizontal dominos.
It follows from the correctness of the algorithm that this gives a maximum domino packing of $P$.
\end{remark}

\subsection{Simpler but slower algorithm}\label{sec:simpler}
Using the structural results on domino packings, we are able to prove the correctness of the following much simpler algorithm, which works by truncating the long edges of $P$.
We sort the corners of $P$ by $x$-coordinates and consider the corners in this order $c_1,\ldots,c_n$.
When $x(c_{i+1})-x(c_i)>9n$, we move all the corners $c_{i+1},\ldots,c_n$ to the left by a distance of $2\lfloor\frac{x(c_{i+1})-x(c_i)}{2}\rfloor-6n$.
We call this operation a \emph{contraction}.
The result after all of the contractions is a polyomino $P'$ with the parities of the $x$-coordinates unchanged and with the difference between the $x$-coordinates of any two consecutive corners at most $6n$.
We then consider the corners in order according to $y$-coordinates and do a similar truncation of the long vertical edges.
We have now reduced the container $P$ to an orthogonal polygon $P''$ of area at most $O(n^4)$, since the span of the $x$-coordinates is $O(n^2)$, as is the span of the $y$-coordinates.
We then compute a maximum matching $M$ in the graph $G(P'')$ and return $|M|+\frac{\area(P)-\area(P'')}{2}$ as the size of a maximum packing in $P$.
Using the multiple-source multiple-sink maximum flow algorithm to compute the matching $M$, this leads to an algorithm with running time $O(n^4\, \text{polylog}\, n)$.

\begin{figure}
\centering
\includegraphics[page=7]{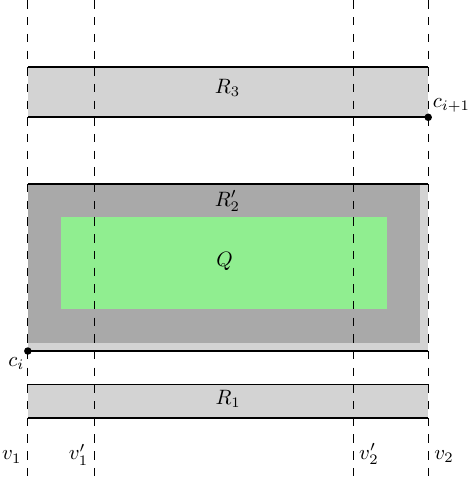}
\caption{A contraction of the simple algorithm with one fat and two skinny rectangles.
The algorithm moves all corners $c_{i+1},\ldots,c_n$ to the left, essentially contracting the area between the vertical lines $v_1'$ and $v_2'$ to nothing.}
\label{figure:simplealg}
\end{figure}

We now verify that the number of uncovered cells in maximum packings is invariant under a single contraction, and the correctness of the algorithm hence follows.
To this end, suppose that $x(c_{i+1})-x(c_i)>9n$, so that we move the corners $c_{i+1},\ldots,c_n$ to the left; see Figure~\ref{figure:simplealg}.
Let $v_1$ and $v_2$ be vertical lines with $x$-coordinates $x(c_i)$ and $x(c_{i+1})$, respectively, and let $V$ be the vertical strip bounded by $v_1$ and $v_2$.
The intersection $P\cap V$ is a collection of disjoint rectangles $R_1,\ldots,R_k$ of width $x(c_{i+1})-x(c_i)$ and various heights.
We define a rectangle $R_i$ to be \emph{fat} if its height is more than $3n$, and otherwise $R_i$ is \emph{skinny}.
We now define a polyomino $P_0$ in order to apply Lemma~\ref{lemma:removingQ}.
For each fat rectangle $R_i$, we let $R'_i\subseteq R_i$ be the maximum rectangle with even coordinates and add $R'_i$ to $P_0$.
As each rectangle $R_i$ corresponds to exactly two horizontal edges, the number of rectangles $k$ is upper bounded by $n/4$ and in particular, the number of corners of $P\setminus P_0$ is at most $2n$.
Letting $Q\mydef B(P_0,-\lfloor 3n/2\rfloor)$, we get from Lemma~\ref{lemma:removingQ} that there exists a maximum tiling of $P$ that restricted to $Q$ is a tiling.

We define $P_1\mydef P\setminus Q$ and observe that the contraction corresponds to contracting a set of long pipes in $P_1$.
These pipes are the skinny rectangles $R_i$ and the parts of the fat rectangles vertically above and below the removed part $Q$.
We therefore get from Lemma~\ref{lemma:pipes} that the number of uncovered cells is invariant under the contraction.

\begin{figure}
\centering
\includegraphics[page=6]{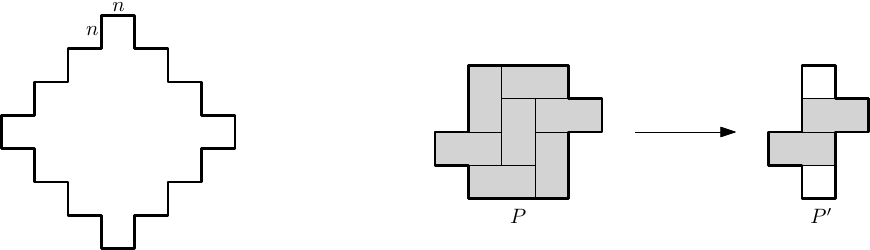}
\caption{Left: A polyomino with area $\Omega(n^4)$ that the simple algorithm will not reduce.
Right: If we truncate edges so that consecutive $x$-coordinates have difference either $1$ or $2$ (keeping the parities invariant), then there may be more uncovered cells in a maximum packing of the reduced instance than in the original.}
\label{figure:counterex}
\end{figure}

For some containers $P$, the graph $G(P'')$ really has $\Omega(n^4)$ vertices, so the algorithm is slower than the complicated algorithm.
For instance when the boundary of $P$ consists of four ``staircases'', each consisting of $n/4$ vertices, where each step has width and height $n$; see Figure~\ref{figure:counterex} (left).
Here the complicated algorithm will remove most of the interior, leaving a layer of cells of thickness $O(n)$ around the boundary, but the simple algorithm will not make any contractions.

One might be tempted to think that we can even truncate the edges so that the difference between consecutive $x$- and $y$-coordinates is either $1$ or $2$, keeping the parity of all coordinates.
However, this does not work, as seen in Figure~\ref{figure:counterex} (right).
Two dominos can be packed in the reduced container $P'$, and the reduction decreases the area by eigth cells, so the formula would give that the original container $P$ has room for six dominos, but there is actually room for seven.

\bibliographystyle{plain}
\bibliography{tiling}

\begin{thebibliography}{10}

\bibitem{beauquier1991translating}
Dani{\`{e}}le Beauquier and Maurice Nivat.
\newblock On translating one polyomino to tile the plane.
\newblock {\em Discrete \& Computational Geometry}, 6:575--592, 1991.

\bibitem{BEAUQUIER19951}
Danièle Beauquier, Maurice Nivat, Eric Remila, and Mike Robson.
\newblock Tiling figures of the plane with two bars.
\newblock {\em Computational Geometry}, 5(1):1--25, 1995.

\bibitem{bergerUndecidability}
Robert Berger.
\newblock The undecidability of the domino problem.
\newblock {\em Memoirs of the American Mathematical Society}, 1(66), 1966.

\bibitem{BERMAN1990153}
Fran Berman, David Johnson, Tom Leighton, Peter~W. Shor, and Larry Snyder.
\newblock Generalized planar matching.
\newblock {\em Journal of Algorithms}, 11(2):153--184, 1990.

\bibitem{berman1982optimal}
Francine Berman, Frank~Thomson Leighton, and Lawrence Snyder.
\newblock Optimal tile salvage, 1982.
\newblock Technical report, Purdue University, Department of Computer Sciences,
  \url{https://docs.lib.purdue.edu/cgi/viewcontent.cgi?article=1321&context=cstech}.

\bibitem{borradaile2017multiple}
Glencora Borradaile, Philip~N. Klein, Shay Mozes, Yahav Nussbaum, and Christian
  Wulff{-}Nilsen.
\newblock Multiple-source multiple-sink maximum flow in directed planar graphs
  in near-linear time.
\newblock {\em {SIAM} Journal on Computing}, 46(4):1280--1303, 2017.

\bibitem{chien2001cutting}
Chen-Fu Chien, Shao-Chung Hsu, and Jing-Feng Deng.
\newblock A cutting algorithm for optimizing the wafer exposure pattern.
\newblock {\em IEEE Transactions on Semiconductor Manufacturing},
  14(2):157--162, 2001.

\bibitem{CONWAY1990183}
J.H Conway and J.C Lagarias.
\newblock Tiling with polyominoes and combinatorial group theory.
\newblock {\em Journal of Combinatorial Theory, Series A}, 53(2):183 -- 208,
  1990.

\bibitem{de2005investigation}
Dirk~K. de~Vries.
\newblock Investigation of gross die per wafer formulas.
\newblock {\em IEEE Transactions on Semiconductor Manufacturing},
  18(1):136--139, 2005.

\bibitem{el2009packing}
Dania El-Khechen, Muriel Dulieu, John Iacono, and Nikolaj Van~Omme.
\newblock Packing $2\times 2$ unit squares into grid polygons is {NP}-complete.
\newblock In {\em Proceedings of the 21st Canadian Conference on Computational
  Geometry (CCCG 2009)}, pages 33--36, 2009.

\bibitem{fowler1981optimal}
Robert~J. Fowler, Michael~S. Paterson, and Steven~L. Tanimoto.
\newblock Optimal packing and covering in the plane are {NP}-complete.
\newblock {\em Information processing letters}, 12(3):133--137, 1981.

\bibitem{gamow1958puzzle}
George Gamow and Marvin Stern.
\newblock {\em Puzzle-math}.
\newblock Macmillan, 1958.

\bibitem{gawrychowski2018improved}
Pawel Gawrychowski and Adam Karczmarz.
\newblock Improved bounds for shortest paths in dense distance graphs.
\newblock In {\em 45th International Colloquium on Automata, Languages, and
  Programming ({ICALP} 2018)}, pages 61:1--61:15, 2018.

\bibitem{doi:10.1080/00029890.1954.11988548}
S.~W. Golomb.
\newblock Checker boards and polyominoes.
\newblock {\em The American Mathematical Monthly}, 61(10):675--682, 1954.

\bibitem{held1994pocket}
Martin Held, G{\'a}bor Luk{\'a}cs, and L{\'a}szl{\'o} Andor.
\newblock Pocket machining based on contour-parallel tool paths generated by
  means of proximity maps.
\newblock {\em Computer-Aided Design}, 26(3):189--203, 1994.

\bibitem{hochbaum1985approximation}
Dorit~S. Hochbaum and Wolfgang Maass.
\newblock Approximation schemes for covering and packing problems in image
  processing and {VLSI}.
\newblock {\em Journal of the ACM (JACM)}, 32(1):130--136, 1985.

\bibitem{horiyama2012packing}
Takashi Horiyama, Takehiro Ito, Keita Nakatsuka, Akira Suzuki, and Ryuhei
  Uehara.
\newblock Packing trominoes is {NP}-complete, {{\#}P}-complete and
  {ASP}-complete.
\newblock In {\em 24th Canadian Conference on Computational Geometry ({CCCG}
  2012)}, pages 211--216, 2012.

\bibitem{JangWafer}
S.~{Jang}, J.~{Kim}, T.~{Kim}, H.~{Lee}, and S.~{Ko}.
\newblock A wafer map yield prediction based on machine learning for
  productivity enhancement.
\newblock {\em IEEE Transactions on Semiconductor Manufacturing},
  32(4):400--407, 2019.

\bibitem{kenyontiling}
C.~{Kenyon} and R.~{Kenyon}.
\newblock Tiling a polygon with rectangles.
\newblock In {\em Proceedings of the 33rd Annual Symposium on Foundations of
  Computer Science (FOCS 1992)}, pages 610--619, 1992.

\bibitem{martinez2010skeleton}
J.~Mart{\'\i}nez, M.~Vigo, N.~Pla-Garc{\'\i}a, and D.~Ayala.
\newblock Skeleton computation of an image using a geometric approach.
\newblock In {\em 31st Eurographics (EG 2010)}, 2010.

\bibitem{melzner2007maximization}
Hanno Melzner and Alexander Olbrich.
\newblock Maximization of good chips per wafer by optimization of memory
  redundancy.
\newblock {\em IEEE Transactions on Semiconductor Manufacturing}, 20(2):68--76,
  2007.

\bibitem{pak2016fast}
Igor Pak, Adam Sheffer, and Martin Tassy.
\newblock Fast domino tileability.
\newblock {\em Discrete \& Computational Geometry}, 56(2):377--394, 2016.

\bibitem{PAK20131804}
Igor Pak and Jed Yang.
\newblock Tiling simply connected regions with rectangles.
\newblock {\em Journal of Combinatorial Theory, Series A}, 120(7):1804 -- 1816,
  2013.

\bibitem{papadopoulou2001voronoi}
Evanthia Papadopoulou and D.T. Lee.
\newblock The {$L_\infty$} {V}oronoi diagram of segments and {VLSI}
  applications.
\newblock {\em International Journal of Computational Geometry \&
  Applications}, 11(05):503--528, 2001.

\bibitem{remilatiling}
Eric R\'{e}mila.
\newblock Tiling a polygon with two kinds of rectangles.
\newblock {\em Discrete Comput. Geom.}, 34(2):313--330, 2005.

\bibitem{thurston1990conway}
William~P. Thurston.
\newblock Conway's tiling groups.
\newblock {\em The American Mathematical Monthly}, 97(8):757--773, 1990.

\bibitem{WIJSHOFF19841}
H.A.G. Wijshoff and J.~{van Leeuwen}.
\newblock Arbitrary versus periodic storage schemes and tessellations of the
  plane using one type of polyomino.
\newblock {\em Information and Control}, 62(1):1--25, 1984.

\end{thebibliography}

\end{document}